%% file: main.tex
\documentclass{llncs}

\newcommand{\macrospath}{macros}

\input{\macrospath/macros.tex}

\renewcommand{\lssym}{\esym}
\renewcommand{\db}{\msym}

\newboolean{withproofs}
\setboolean{withproofs}{true}
\begin{document}

\mainmatter

\title{A Strong Distillery}

\author{
Beniamino Accattoli\inst{1} \and
Pablo Barenbaum\inst{2} \and
Damiano Mazza\inst{3}
}

\institute{
INRIA, UMR 7161, LIX, \'Ecole Polytechnique, CNRS\\
\email{beniamino.accattoli@inria.fr}
\and
University of Buenos Aires -- CONICET\\
\email{pbarenbaum@dc.uba.ar}
\and
CNRS, UMR 7030, LIPN, Universit\'e Paris 13, Sorbonne Paris Cit\'e \\
\email{Damiano.Mazza@lipn.univ-paris13.fr}
}

\maketitle

	\input{00_-_Abstract}

\input{01_-_Introduction}
\input{02_-_Linear_Leftmost-Outermost_Reduction}
\input{03_-_Distilleries}
\input{04_-_Strengthening_the_MAM}
\input{05_-_The_Strong_MAM}
\input{06_-_Distilling_the_Strong_MAM}
\input{07_-_Complexity_Analysis}

\subsubsection*{Acknowledgments.}
\footnotesize This work was partially supported by projects \textsc{Logoi} ANR-2010-BLAN-0213-02, \textsc{Coquas} ANR-12-JS02-006-01, \textsc{Elica} ANR-14-CE25-0005,
the Saint-Exup\'ery program funded by the French embassy and the Ministry of Education in Argentina, the French--Argentinian laboratory in Computer Science INFINIS, and the French Argentinian project ECOS-Sud A12E04.
\normalsize

\bibliographystyle{splncs03}
\bibliography{\macrospath/biblio}

\input{appendix.tex}

\end{document}

%% file: macros/macros.tex
\usepackage[utf8]{inputenc}

\usepackage{amsmath}
\usepackage{amssymb}
\usepackage{stmaryrd}
\usepackage{bussproofs}

\usepackage{ifthen}
\usepackage{xspace}
\usepackage{fancybox}
\usepackage{enumitem}

\input{\macrospath/ben-macros-diagrams.tex}

\newcommand{\ignore}[1]{}
\newcommand{\sep}{\hspace*{0.5cm}}
\newcommand{\ms}{\medskip}

\newcommand{\myinput}[1]{\ifthenelse{\boolean{withimages}}{\input{#1}}{}}

\newcommand{\reflemma}[1]{Lemma~\ref{l:#1}}
\newcommand{\reflemmap}[2]{Lemma~\ref{l:#1}.\ref{p:#1-#2}}
\newcommand{\reflemmapp}[3]{Lemma~\ref{l:#1}.\ref{p:#1-#2}.\ref{p:#1-#2-#3}}

\newcommand{\reflemmaeq}[1]{{L.\ref{l:#1}}}
\newcommand{\reflemmaeqp}[2]{{L.\ref{l:#1}.\ref{p:#1-#2}}}
\newcommand{\refpoint}[1]{Point~\ref{p:#1}}
\newcommand{\refpointeq}[1]{P.\ref{p:#1}}

\newcommand{\refpointmute}[1]{\ref{p:#1}}

\newcommand{\refthm}[1]{Theorem~\ref{thm:#1}}

\newcommand{\reftm}[1]{Theorem~\ref{tm:#1}}

\newcommand{\refprop}[1]{Proposition~\ref{prop:#1}}
\newcommand{\refsect}[1]{Sect.~\ref{sect:#1}}

\newcommand{\refapp}[1]{Appendix~\ref{app:#1}}

\newcommand{\reffig}[1]{Fig.~\ref{fig:#1}}

\newcommand{\refdefp}[2]{Definition~\ref{def:#1}.\ref{p:#1-#2}}

\renewcommand{\case}[1]{{\bf Case #1.}}
\newcommand{\casealt}[1]{{\bf #1.}}
\newcommand{\caselight}[1]{\textit{#1}}

\newcommand{\ie}{\textit{i.e.}\xspace}
\newcommand{\eg}{\textit{e.g.}\xspace}
\newcommand{\ih}{\textit{i.h.}\xspace}

\newcommand{\defeq}{\mathrel{:=}}
\newcommand{\grameq}{\mathrel{::=}}
\newcommand{\set}[1]{\{#1\}}

\newcommand{\size}[1]{|#1|}

\newcommand{\LeftRightarrow}{\Lleftarrow\!\!\!\!\Rrightarrow}

\newcommand{\db}{{\tt dB}}

\newcommand{\lssym}{{\tt ls}}

\newcommand{\admsym}{{\mathtt c}}

\newcommand{\lsc}{LSC}

\renewcommand{\l}{\lambda}
\newcommand{\isub}[2]{\{#1/#2\}}
\renewcommand{\isub}[2]{\{#1{\shortleftarrow}#2\}}
\newcommand{\esub}[2]{[#1/#2]}
\renewcommand{\esub}[2]{[#1{\shortleftarrow}#2]}
\newcommand{\fv}[1]{{\tt fv}(#1)}

\newcommand{\varsplit}[3]{#1_{[#3]_{#2}}}

\newcommand{\rootRew}[1]{\mapsto_{#1}}
\newcommand{\Rew}[1]{\rightarrow_{#1}}

\newcommand{\rtodb}{\rootRew{\db}}
\newcommand{\rtols}{\rootRew{\lssym}}

\newcommand{\rtom}{\rootRew{\msym}}
\newcommand{\rtoe}{\rootRew{\esym}}

\newcommand{\tostructsym}{\LeftRightarrow}

\newcommand{\esym}{{\mathtt e}}
\newcommand{\msym}{{\mathtt m}}

\newcommand{\psym}{{\mathtt p}}

\newcommand{\togen}{\multimap}

\newcommand{\tom}{\Rew{\msym}}
\newcommand{\toe}{\Rew{\esym}}

\newcommand{\eqstruct}{\equiv}
\newcommand{\tostruct}{\eqstruct}
\newcommand{\tostructgc}{\tostruct_{\structrulegc}}

\newcommand{\tostructdup}{\tostruct_{\structruledup}}
\newcommand{\tostructlam}{\tostruct_{\lambda}}

\newcommand{\tostructapl}{\tostruct_{\structruleapl}}
\newcommand{\tostructapr}{\tostruct_{\structruleapr}}
\newcommand{\tostructes}{\tostruct_{[\cdot]}}
\newcommand{\tostructcom}{\tostruct_{\structrulecom}}
\newcommand{\structrulegc}{\mathtt{gc}}

\newcommand{\structruledup}{\mathtt{dup}}

\newcommand{\structruleapl}{\mathtt{@l}}
\newcommand{\structruleapr}{\mathtt{@r}}

\newcommand{\structrulecom}{\mathtt{com}}

\newcommand{\tm}{t}
\newcommand{\tmtwo}{u}
\newcommand{\tmthree}{w}
\newcommand{\tmfour}{r}
\newcommand{\tmfive}{q}

\newcommand{\tmp}{\tm'}
\newcommand{\tmtwop}{\tmtwo'}
\newcommand{\tmthreep}{\tmthree'}
\newcommand{\tmfourp}{\tmfour'}
\newcommand{\tmfivep}{\tmfive'}

\newcommand{\tmpp}{\tm''}
\newcommand{\tmtwopp}{\tmtwo''}

\newcommand{\var}{x}
\newcommand{\vartwo}{y}
\newcommand{\varthree}{z}

\newcommand{\ctxholep}[1]{\langle #1\rangle}
\newcommand{\ctxhole}{\ctxholep{\cdot}}

\newcommand{\ctx}{C}
\newcommand{\ctxtwo}{C'}
\newcommand{\ctxthree}{C''}
\newcommand{\ctxfour}{C'''}
\newcommand{\ctxp}[1]{\ctx\ctxholep{#1}}
\newcommand{\ctxtwop}[1]{\ctxtwo\ctxholep{#1}}
\newcommand{\ctxthreep}[1]{\ctxthree\ctxholep{#1}}

\newcommand{\sctx}{L}
\newcommand{\sctxtwo}{\sctx'}
\newcommand{\sctxthree}{\sctx''}
\newcommand{\sctxp}[1]{\sctx\ctxholep{#1}}
\newcommand{\sctxtwop}[1]{\sctxtwo\ctxholep{#1}}
\newcommand{\sctxthreep}[1]{\sctxthree\ctxholep{#1}}

\newcommand{\arbctxp}[1]{\arbctxp{#1}}
\newcommand{\arbctxtwop}[1]{\arbctxtwop{#1}}

\newcommand{\commone}{\skeval\csym_1}
\newcommand{\commtwo}{\skeval\csym_2}
\newcommand{\commthree}{\skeval\csym_3}
\newcommand{\commfour}{\skback\csym_4}
\newcommand{\commfive}{\skback\csym_5}
\newcommand{\commsix}{\skback\csym_6}

\newcommand{\tomachhole}[1]{\leadsto_{#1}}
\newcommand{\tomach}{\tomachhole{}}
\newcommand{\tomachm}{\tomachhole{\msym}}

\newcommand{\tomache}{\tomachhole{\esym}}

\newcommand{\tomacha}{\tomachhole{\admsym}}

\newcommand{\tomachasub}[1]{\tomachhole{\admsym_{#1}}}

\newcommand{\tomachasubone}{\tomachhole\commone}
\newcommand{\tomachasubtwo}{\tomachhole\commtwo}
\newcommand{\tomachasubthree}{\tomachhole\commthree}
\newcommand{\tomachasubfour}{\tomachhole\commfour}
\newcommand{\tomachasubfive}{\tomachhole\commfive}
\newcommand{\tomachasubsix}{\tomachhole\commsix}

\newcommand{\admnf}[1]{\mathtt{nf}_{\admsym}(#1)}
\newcommand{\tomachc}{\tomachhole{\admsym}}
\newcommand{\tomachcp}[1]{\tomachhole{\admsym#1}}

\newcommand{\tomachp}{\tomachhole{\psym}}

\newcommand{\code}{\overline{\tm}}
\newcommand{\codetwo}{\overline{\tmtwo}}
\newcommand{\codethree}{\overline{\tmthree}}
\newcommand{\codefour}{\overline{\tmfour}}

\newcommand{\genv}{E}
\newcommand{\genvtwo}{E'}
\newcommand{\genvthree}{E''}

\newcommand{\genvweak}{\genv_{\mathrm w}}
\newcommand{\genvweaktwo}{\genvtwo_{\mathrm w}}

\newcommand{\genvtrunk}{\genv_{\mathrm t}}

\newcommand{\stctx}[1]{\ctx_{#1}}
\newcommand{\stctxp}[2]{\ctx_{#1}\ctxholep{#2}}

\newcommand{\stempty}{\epsilon}
\newcommand{\cons}{:}

\newcommand{\stack}{\pi}
\newcommand{\stacktwo}{\pi'}

\newcommand{\state}{s}
\newcommand{\statetwo}{s'}
\newcommand{\statethree}{s''}

\newcommand{\rename}[1]{#1^\alpha}

\newcommand{\exec}{\rho}
\newcommand{\exectwo}{\sigma}

\newcommand{\decode}[1]{\llbracket #1\rrbracket}
\renewcommand{\decode}[1]{\underline{#1}}

\newcommand{\dstackp}[1]{\decode{\stack}\ctxholep{#1}}

\newcommand{\deriv}{d}
\newcommand{\derivtwo}{e}

\newcommand{\sizee}[1]{|#1|_{\esym}}
\newcommand{\sizem}[1]{|#1|_{\msym}}

\newcommand{\sizep}[1]{|#1|_p}

\newcommand{\distil}{{\tt D}}
\newcommand{\calculus}{{\tt C}}
\renewcommand{\calculus}{\togen}
\newcommand{\mach}{{\tt M}}

\renewcommand{\rtodb}{\rootRew{\msym}}
\renewcommand{\rtols}{\rootRew{\esym}}

\newcommand{\csym}{{\mathtt c}}

\newcommand{\la}[1]{\lambda #1.}

\newcommand{\maca}[5]{#2 & #3 & #4 & #5 & #1}

\newcommand{\mac}[5]{(#2,#3,#4,#5,#1)}

\newcommand{\sizecom}[1]{|#1|_c}

\newcommand{\decodep}[2]{\decode{#1}\ctxholep{#2}}

\newcommand{\myproof}[1]{
\ifthenelse{\boolean{omitproofs}}{\begin{IEEEproof} Proof available but omitted for readability. \end{IEEEproof}}{#1}}

\newcommand{\esmeas}[1]{|#1|_{[\cdot]}}

\newcommand{\decodefun}{\decode{{ }\cdot{ }}}

\newcommand{\node}{\mathtt{n}}

\newcommand{\quiet}{neutral\xspace}
\newcommand{\Quiet}{Neutral\xspace}

\newcommand{\prefix}{\prec_p}
\newcommand{\outin}{\prec_O}
\newcommand{\leftright}{\prec_L}
\newcommand{\leftout}{\prec_{\lo}}

\newcommand{\lo}{LO\xspace}

\newcommand{\ilo}{i\lo}

\newcommand{\tolo}{\Rew{\tiny \mathtt{LO}}}

\newcommand{\lfv}[1]{{\tt lfv}(#1)}

\newcommand{\skamstate}[5]{#2\mid#3\mid#4\mid#5\mid#1}
\newcommand{\skeval}{\textcolor{blue}{\blacktriangledown}}
\newcommand{\skback}{\textcolor{red}{\blacktriangle}}

\newcommand{\compatible}{\propto}
\newcommand{\skap}[2]{(#1,#2)}

\newcommand{\skframe}{F}

\newcommand{\skframeweak}{\skframe_{\mathrm w}}
\newcommand{\skframeweaktwo}{\skframetwo_{\mathrm w}}
\newcommand{\skframetrunk}{\skframe_{\mathrm t}}

\newcommand{\skframetwo}{\skframe'}
\newcommand{\skframethree}{\skframe''}

\newcommand{\skboundvar}{\skeval}
\newcommand{\openscope}{\skeval}
\newcommand{\openscopem}[1]{\skeval #1}
\newcommand{\closescope}{\skback}
\newcommand{\closescopem}[1]{\skback#1}
\newcommand{\skphase}{\varphi}

\newcommand{\skl}{\Lambda}
\newcommand{\sklp}[1]{\Lambda(#1)}

\newcommand{\skammachine}{Strong MAM\xspace}
\newcommand{\rmeasure}[1]{\size{#1}}

\newcommand\scalemath[2]{\medskip\scalebox{#1}{\mbox{\ensuremath{\displaystyle #2}}}\medskip}

\newcommand{\decodecompat}[2]{\decode{#1 \compatible #2}}

\newcommand{\sizecomev}[1]{|#1|_{\skeval \csym}}
\newcommand{\sizecombt}[1]{|#1|_{\skback \csym}}
\newcommand{\polsize}[2]{|#1|_{#2}}

\newcommand{\withproofs}[1]{\ifthenelse{\boolean{withproofs}}{#1}{}}

\newcommand{\withoutproofs}[1]{\ifthenelse{\boolean{withproofs}}{}{#1}}



%% file: macros/ben-macros-diagrams.tex
\usepackage{tikz}
\usetikzlibrary{matrix,arrows}
\usetikzlibrary{decorations.shapes}
\usetikzlibrary{decorations.text}
\usetikzlibrary{decorations.pathmorphing}
\usetikzlibrary{decorations.markings}

\tikzset{
node distance=1.3cm, auto,
every node/.style={font=\tiny },
ocenter/.style={baseline={([yshift=-.5ex, xshift=-.5ex]current bounding box)}},  
labelBeginAbove/.style={postaction={decorate,decoration={markings,mark=at position 0 with {\node[inner sep= 0.6pt, above=1pt]{\tiny #1};}} } },
labelBeginBelow/.style={postaction={decorate,decoration={markings,mark=at position 0 with {\node[inner sep= 0.6pt, below=1pt]{\tiny #1};}}}},
labelEndAbove/.style={postaction={decorate,decoration={markings,mark=at position 1 with {\node[inner sep= 0.6pt, above=1pt]{\tiny #1};}}}},
labelEndBelow/.style={postaction={decorate,decoration={markings,mark=at position 1 with {\node[inner sep= 0.6pt, below=1pt]{\tiny #1};}}}},
labelEndRight/.style={postaction={decorate,decoration={markings,mark=at position 1 with {\node[inner sep= 0.6pt, right=1pt]{\tiny #1};}}}},
labelEndLeft/.style={postaction={decorate,decoration={markings,mark=at position 1 with {\node[inner sep= 0.6pt, left=1pt]{\tiny #1};}}}}
}

\newcommand{\nodeHorDist}{2cm}
\newcommand{\nodeVerDist}{0.8cm}

\newcommand{\commutesredgenEEABCD}[8]{
      \begin{tikzpicture}[ocenter]
       \node (s) {\normalsize\ensuremath{#3}};
       \node at (s.center)  [right=1.7*\nodeHorDist](s1){\normalsize\ensuremath{#4}};
       \node at (s.center)  [below=\nodeVerDist](s2) {\normalsize\ensuremath{#5}};
       \node at (s2.center) [right=1.7*\nodeHorDist](t) {\normalsize\ensuremath{#6}};
       \node at (s.center)  [below=\nodeVerDist/2](eq1){\normalsize\ensuremath{#1}};
       \node at (s1.center) [below=\nodeVerDist/2](eq2){\normalsize\ensuremath{#2}};
       \draw[-o] (s) to node {$#7$} (s1);
       \draw[-o, dashed] (s2) to node {$#8$} (t);
      \end{tikzpicture} 
}

\newcommand{\commutesredEEABCD}[6]{
\commutesredgenEEABCD{#1}{#2}{#3}{#4}{#5}{#6}{}{}
}

\newcommand{\commutesdbEEABCD}[6]{
\commutesredgenEEABCD{#1}{#2}{#3}{#4}{#5}{#6}{\db}{\db}
}

\newcommand{\commuteslsEEABCD}[6]{
\commutesredgenEEABCD{#1}{#2}{#3}{#4}{#5}{#6}{\lssym}{\lssym}
}

%% file: 00_-_Abstract.tex
\begin{abstract}
	Abstract machines for the strong evaluation of $\lambda$-terms (that is, under abstractions) are a mostly neglected topic, despite their use in the implementation of proof assistants and higher-order logic programming languages. This paper introduces a machine for the simplest form of strong evaluation, leftmost-outermost (call-by-name) evaluation to normal form, proving it correct, complete, and bounding its overhead. Such a machine, deemed \emph{Strong Milner Abstract Machine}, is a variant of the KAM computing normal forms and using just one global environment. Its properties are studied via a special form of decoding, called a \emph{distillation}, into the Linear Substitution Calculus, neatly reformulating the machine as a standard micro-step strategy for explicit substitutions, namely \emph{linear leftmost-outermost reduction}, \ie the extension to normal form of linear head reduction. Additionally, the overhead of the machine is shown to be linear both in the number of steps and in the size of the initial term, validating its design. The study highlights two distinguished features of strong machines, namely backtracking phases and their interactions with abstractions and environments.
\end{abstract}


%% file: 01_-_Introduction.tex
\section{Introduction}
The computational model behind functional programming is the weak $\l$-calculus, where \emph{weakness} is the fact that evaluation stops as soon as an abstraction is obtained. Evaluation is usually defined in a small-step way, specifying a strategy for the selection of weak $\beta$-redexes. Both the advantage and the drawback of $\l$-calculus is the lack of a machine in the definition of the model. Unsurprisingly implementations of functional languages have been explored for decades. 

Implementation schemes are called \emph{abstract machines}, and usually account for two tasks. First, they switch from small-step to \emph{micro-step} evaluation, delaying the costly meta-level substitution used in small-step operational semantics and replacing it with substitutions of one occurrence at a time, when required. Second, they also  \emph{search the next redex} to reduce, walking through the program according to some evaluation strategy. Abstract machines are \emph{machines} because they are deterministic and the complexity of their steps can easily be measured, and are \emph{abstract} because they omit many details of a real implementation, like the actual representation of terms and data-structures or the garbage collector.

Historically, the theory of $\l$-calculus and the implementation of functional languages have followed orthogonal approaches. The former rather dealt with \emph{strong} evaluation, and it is only since the seminal work of Abramsky and Ong \cite{DBLP:journals/iandc/AbramskyO93} that the theory took weak evaluation seriously. Dually, practical studies mostly ignored  \emph{strong} evaluation, with the notable exception of Cr\'{e}gut \cite{DBLP:conf/lfp/Cregut90,DBLP:journals/lisp/Cregut07} (1990) and, more recently, the semi-strong approach of Gr\'egoire and Leroy \cite{DBLP:conf/icfp/GregoireL02} (2002)---see also the \emph{related work} paragraph below. Strong evaluation is nonetheless essential in the implementation of proof assistants or higher-order logic programming, typically for type-checking in frameworks with dependent types as the Edinburgh Logical Framework or the Calculus of Constructions, as well as for unification modulo $\beta\eta$ in simply typed frameworks like $\l$-prolog. 

The aim of this paper is to move the first steps towards a systematic and theoretical exploration of the implementation of strong evaluation. Here we deal with the simplest possible case, call-by-name evaluation to strong normal form, implemented by a variant of the Krivine Abstract Machine. The study is carried out according to the \emph{distillation methodology}, a new approach recently introduced by the authors and previously applied only to weak evaluation \cite{DBLP:conf/icfp/AccattoliBM14}.

\paragraph*{Distilling Abstract Machines.}
Many abstract machines can be rephrased as strategies in \emph{$\l$-calculi with explicit substitutions} (ES for short), see at least \cite{DBLP:journals/tcs/Curien91,DBLP:journals/jfp/HardinM98,DBLP:journals/lisp/Cregut07,DBLP:journals/tocl/BiernackaD07,DBLP:journals/lisp/Lang07,DBLP:journals/toplas/AriolaBS09}. The Linear Substitution Calculus (LSC)---a variation over a $\l$-calculus with ES by Robin Milner \cite{DBLP:journals/entcs/Milner07} developed by Accattoli and Kesner \cite{DBLP:conf/rta/Accattoli12,DBLP:conf/popl/AccattoliBKL14}---provides more than a simple reformulation: it disentangles the two tasks carried out by abstract machines, retaining the \emph{micro-step operational semantics} and omitting the \emph{search for the next redex}. Such a neat disentangling, that we prefer to call a \emph{distillation}, is a decoding based on the following key points:
\begin{enumerate}
 \item \emph{Partitioning}: the machine transitions are split in two classes. \emph{Principal transitions} are mapped to the rewriting rules of the calculus, while \emph{commutative transitions}---responsible for the search for the redex---are mapped on a notion of structural equivalence, specific to the LSC.
 \item \emph{Rewriting}: structural equivalence accounts both for the search for the redex and garbage collection, and commutes with evaluation. It can thus be postponed, isolating the micro-step strategy in the rewriting of the LSC.
 \item \emph{Logic}: the LSC itself has only two rules, corresponding to cut-elimination in linear logic proof nets. A distillation then provides a logical reading of an abstract machine (see \cite{DBLP:conf/icfp/AccattoliBM14} for more details).
 \item \emph{Complexity}: by design, a principal transition has to take linear time in the input, while a commutative transition has to be constant. 
\end{enumerate}
 
A \emph{distillery} is then given by a machine, a strategy, a structural equivalence, and a decoding function satisfying the above points. In bilinear distilleries, the number of commutative transitions is linear in both the \emph{number of principal transitions} and the \emph{size of the initial term}. Bilinearity guarantees that distilling away the commutative part by switching to the LSC preserves the asymptotical behavior, \ie\ it does not forget too much. At the same time, the bound on the commutative overhead justifies the design of the abstract machine, providing a provably bounded implementation scheme.
 
\paragraph*{A Strong Distillery.}
Our machine is a strong version of the Milner Abstract Machine (MAM), a variant with just one \emph{global environment} of the Krivine Abstract Machine (KAM), introduced in \cite{DBLP:conf/icfp/AccattoliBM14}. 

The first result of the paper is the design of a distillery relating the \skammachine to \emph{linear leftmost-outermost reduction} in the LSC \cite{DBLP:conf/popl/AccattoliBKL14,DBLP:conf/csl/AccattoliL14}---that is at the same time a refinement of leftmost-outermost (\lo) $\beta$-reduction and an extension of linear head reduction \cite{DBLP:journals/tcs/MascariP94,Danos04headlinear,DBLP:conf/rta/Accattoli12} to normal form---together with the proof of correctness and completeness of the implementation \cite{DBLP:journals/tcs/Plotkin75}. Moreover, the linear \lo strategy is \emph{standard} and \emph{normalizing} \cite{DBLP:conf/popl/AccattoliBKL14}, and thus we provide an instance of Plotkin's approach of mapping abstract machines to such strategies \cite{DBLP:conf/popl/AccattoliBKL14}.

The second result is the complexity analysis showing that the distillery is bilinear, \ie that  the cost of the additional search for the next redex specific to the machine is negligible.  The analysis is simple, and yet subtle and robust. It is subtle because it requires a global analysis of executions, and it is robust because the overhead is bilinear for \emph{any} evaluation sequence, not necessarily to normal form, and even for diverging ones. 

For the design of the \skammachine we make various choices:
\begin{enumerate}
 \item \emph{Global Environment}: we employ a \emph{global} environment, which is in opposition to having closures (pairing subterms with \emph{local} environments), and it models a store-based implementation scheme. The choice is motivated by future extensions to more efficient strategies as call-by-need, where the global environment allows to integrate sharing with a form of memoization \cite{DBLP:conf/ppdp/DanvyZ13,DBLP:conf/icfp/AccattoliBM14}.

 \item \emph{Sequential Exploration and Backtracking}: we fix a sequential exploration of the term (according to the leftmost-outermost order), in opposition to the parallel evaluation of the arguments (once a head normal form has been reached). This choice internalizes the handling of the recursive iterations, that would be otherwise left to the meta-level, providing a finer study of the data-structures needed by a strong machine. 
 On the other hand, it forces to have backtracking transitions, activated when the current subterm has been checked to be normal and evaluation needs to retrieve the next subterm on the stack. Call-by-value machines usually have a similar but simpler backtracking mechanism, realized via an additional component, the \emph{dump}. 

 \item \emph{(Almost) No Garbage Collection}: we focus on time complexity, and thus ignore space issues, that is, our machine does not account for garbage collection. In particular, we keep the global environment completely unstructured, similarly to the (weak) MAM. Strong evaluation however is subtler, as to establish a precise relationship between the machine and the calculus with ES, garbage collection cannot be completely ignored. Our approach is to isolate it within the meta-level: we use a system of parenthesized markers, to delimit subenvironments created under abstractions that could be garbage collected once the machine backtracks outside those abstraction. These labels are not inspected by the transitions, and play a role only for the proof of the distillation theorem. Garbage collection then is somewhat accounted for by the analysis, but there are no dedicated transitions nor rewriting rules, it is rather encapsulated in the decoding and in the structural equivalence.
 \end{enumerate}


\paragraph*{Efficiency?}
It is known that \lo\ evaluation is not efficient. Improvements are possible along three axis: refining the strategy (by turning to strong call-by-value/need, partially done in \cite{DBLP:conf/icfp/GregoireL02,DBLP:journals/lisp/Cregut07,fireballs}), speeding up the substitution process (by forbidding the substitution of variables, see \cite{DBLP:conf/wollic/AccattoliC14,fireballs}), and  avoiding useless substitutions (by adding \emph{useful sharing}, see \cite{DBLP:conf/csl/AccattoliL14,fireballs}). These improvements however require sophisticated machines, left to future work. 

\lo\ evaluation is nonetheless a good first case study, as it allows to isolate the analysis of backtracking phases and their subtle interactions with abstractions and environments. We expect that the mentioned optimizations can be added in a quite modular way, as they  have all been addressed in the complementary study in \cite{fireballs}, based on the same technology (\ie LSC and distilleries).

\paragraph*{(Scarce) Related Work.}
Beyond Cr\'{e}gut's \cite{DBLP:conf/lfp/Cregut90,DBLP:journals/lisp/Cregut07}, we are aware of only two other similar works on strong abstract machines, 
Garc{\'{\i}}a{-}P{\'{e}}rez, Nogueira and Moreno-Navarro's \cite{DBLP:conf/ppdp/Garcia-PerezNM13} (2013),
and Smith's \cite{ConnorSmith2014} (unpublished, 2014). 
Two further studies, de Carvalho's \cite{DBLP:journals/corr/abs-0905-4251} and  
Ehrhard and Regnier's \cite{DBLP:conf/cie/EhrhardR06}, introduce strong versions of the KAM but for theoretical purposes; in particular, their design choices are not tuned towards implementations (\eg rely on a na\"ive parallel exploration of the term). 
Semi-strong machines for call-by-value (\ie dealing with weak evaluation but on open terms) are studied by Gr\'egoire and Leroy \cite{DBLP:conf/icfp/GregoireL02} and in a recent work by Accattoli and Sacerdoti Coen \cite{fireballs} (see \cite{fireballs} for a comparison with \cite{DBLP:conf/icfp/GregoireL02}). 
More recent work by D\'en\`es \cite{Denes} and Boutiller \cite{Boutiller} appeared in the context of term evaluation in Coq. These works, which do offer the nice perspective of concretely dealing with proof assistants, are focused on quite specific Coq-related tasks (such as term simplification) and the difference in reduction strategy and underlying motivations makes a comparison difficult.

Of all the above, the closest to ours is Cr\'egut's work, because it defines an implementation-oriented strong KAM, thus also addressing leftmost-outermost reduction. His machine uses local environments, sequential exploration and backtracking, scope markers akin to ours, and a calculus with ES to establish the correctness of the implementation. His calculus, however, has no less than 13 rewriting rules, while ours just 2, and so our approach is simpler by an order of magnitude. Moreover, we want to stress that our contribution does not lie in the machine \emph{per se}, or the chosen reduction strategy (as long as it is strong), but in the combined presence of a robust and simple abstraction of the machine, provided by the LSC, and the complexity analysis showing that such an abstraction does not miss too much.  In this respect, none of the above works comes with an analysis of the overhead of the machine nor with the logical and rewriting perspective we provide.  In fact, our approach offers general guidelines for the design of (strong) abstract machines. The choice of leftmost-outermost reduction showcases the idea while keeping technicalities to a minimum, but it is by no means a limitation. The development of strong distilleries for call-by-value or lazy strategies, which may be more attractive from a programming languages perspective, are certainly possible and will be the object of future work (again, an intermediary step has already been taken in \cite{fireballs}).

Global environments are explored by Fern{\'{a}}ndez and Siafakas in \cite{DBLP:journals/entcs/FernandezS09}, and used in a minority of works, \eg \cite{DBLP:conf/birthday/SandsGM02,DBLP:conf/ppdp/DanvyZ13}. We introduced the distillation technique in \cite{DBLP:conf/icfp/AccattoliBM14} to revisit the relationship between the KAM and weak linear head reduction pointed out by Danos and Regnier \cite{Danos04headlinear}. Distilleries have also been used in \cite{fireballs}. The idea to distinguish between \emph{operational content} and \emph{search for the redex} in an abstract machine is not new, as it underlies in particular the \emph{refocusing semantics} of Danvy and Nielsen \cite{Danvy04refocusingin}. The LSC, with its roots in linear logic proof nets, allows to see this distinction as an avatar of the principal/commutative divide in cut-elimination, because machine transitions may be seen as cut-elimination steps \cite{DBLP:journals/toplas/AriolaBS09,DBLP:conf/icfp/AccattoliBM14}. Hence, it is fair to say that distilleries bring an original refinement where logic, rewriting, and complexity enlighten the picture, leading to formal bounds on machine overheads.

Omitted proofs may be found in the appendices.

%% file: 02_-_Linear_Leftmost-Outermost_Reduction.tex
\section{Linear Leftmost-Outermost Reduction}
\label{sect:lo-reduction}
The language of the \emph{linear substitution calculus} (\lsc\ for short) is given by the following term grammar:
\[\begin{array}{c@{\hspace{2em}}ccccc}
   \mbox{LSC Terms} & \tm,\tmtwo,\tmthree,\tmfour & \grameq & \var\mid \l \var. \tm \mid \tm \tmtwo\mid  \tm\esub\var\tmtwo.
  \end{array}\]
The constructor $\tm\esub{\var}{\tmtwo}$ is called an \emph{explicit
  substitution, shortened ES} (of $\tmtwo$ for $\var$ in $\tm$). Both $\l
\var. \tm$ and $\tm\esub{\var}{\tmtwo}$ bind $\var$ in $\tm$, and we
silently work modulo $\alpha$-equivalence of these bound variables,
\eg\ $(\var\vartwo)\esub\vartwo\tm\isub\var\vartwo =
(\vartwo\varthree)\esub\varthree\tm$. 

The operational semantics of the \lsc\ is parametric
in a notion of (one-hole) context. General \emph{contexts}, that simply extend the contexts for $\l$-terms with the two cases for ES, and the special case of \emph{substitution contexts} are defined by:
\[\begin{array}{r@{\hspace{2em}}cclcc}
   \mbox{Contexts} & \ctx, \ctxtwo & \grameq & \ctxhole\mid \l \var. \ctx\mid \ctx \tm \mid\tm\ctx\mid\ctx\esub{\var}{\tm}\mid\tm\esub{\var}{\ctx};\\
   \mbox{Substitution Contexts} & \sctx, \sctxtwo & \grameq & \ctxhole \mid \sctx\esub\var\tm.
  \end{array}\]
We write $\ctx\prefix\tm$ if there is a term $\tmtwo$ s.t. $\ctxp\tmtwo = \tm$, call it the
\emph{prefix relation}.

The rewriting relation is $\to \defeq \tom\cup\toe$ where $\tom$ and $\toe$ are the \emph{multiplicative} and \emph{exponential} rules, defined by
\[\begin{array}{r@{\hspace{0.8cm}}c@{\hspace{1cm}}c}
  &\textsc{Rule at Top Level} & \textsc{Contextual closure} \\
  \mbox{Multiplicative} & \sctxp{\l \var.\tm} \tmtwo  \rtom \sctxp{\tm\esub{\var}{\tmtwo}} &
        \ctxp \tm \tom \ctxp \tmtwo \textrm{~~~if } \tm \rtom \tmtwo \\

  \mbox{Exponential} & \ctxp\var \esub\var\tmtwo \rtoe \ctxp\tmtwo \esub\var\tmtwo &
        \ctxp \tm \toe \ctxp \tmtwo \textrm{~~~if } \tm \rtoe \tmtwo
\end{array}\]

The rewriting
rules are assumed to use \emph{on-the-fly} $\alpha$-equivalence to
avoid variable capture. For instance, $(\la \var\tm)\esub\vartwo\tmtwo \vartwo  \tom
   \tm\isub\vartwo\varthree\esub\var\vartwo\esub\varthree\tmtwo$ for
     $\varthree\notin\fv\tm$, and $(\la\vartwo(\var\vartwo))\esub{\var}{\vartwo} \toe (\la
   \varthree(\vartwo\varthree))\esub{\var}{\vartwo}$. 
Moreover, in $\toe$ the context $\ctx$ is assumed to not capture
$\var$, in order to have $(\l \var.\var)\esub{\var}{\vartwo} \not\toe (\l
\var.\vartwo)\esub{\var}{\vartwo}$.

The above operational semantics ignores garbage collection. In the \lsc, this may be realized by an additional rule which may always be postponed, see \cite{DBLP:conf/rta/Accattoli12}. 

Taking the external context into account, an exponential step has the
form $
\ctxtwop{\ctxp{\var}\esub{\var}{\tmtwo}} \toe
\ctxtwop{\ctxp{\tmtwo}\esub{\var}{\tmtwo}}$. We shall often use a
\emph{compact} form:
\[\begin{array}{c@{\hspace{1.5cm}}c}
  \multicolumn{2}{c}{\textsc{Exponential Rule in Compact Form}}\\
	\ctxthreep{\var}\toe\ctxthreep{\tmtwo} &
        \mbox{if }\ctxthree =
\ctxtwop{\ctx\esub{\var}{\tmtwo}}\\

\end{array}\]

\begin{definition}[Redex Position]
Given a $\tom$-step $\ctxp\tm \tom \ctxp\tmtwo$ with $\tm\rtom\tmtwo$ or a compact $\toe$-step $\ctxp\var\toe\ctxp\tm$, the \emph{position} of the redex is the context $\ctx$.
\end{definition}

We identify a redex with its position, thus using $\ctx,\ctxtwo,\ctxthree$ for redexes, and use $\deriv:\tm\to^k\tmtwo$ for derivations, \ie\ for possibly empty sequences of rewriting steps. We write $\esmeas\tm$ for
the number of substitutions in $\tm$, and use $\size\deriv$, $\sizem\deriv$, and $\sizee\deriv$ for the number of steps, $\msym$-steps, and $\esym$-steps in  $\deriv$, respectively.

\paragraph{Linear Leftmost-Outermost Reduction, Two Definitions.} 
We give two definitions of linear \lo reduction $\tolo$, a traditional one based on ordering redexes and a new contextual one not mentioning the order, apt to work with LSC and relate it to abstract machines. We start by defining the \lo\ order on contexts.

\begin{definition}[\lo\ Order]  
    The \emph{outside-in order} $\ctx\outin\ctxtwo$ is defined by
    \begin{enumerate}
    \item 
      \emph{Root}: $\ctxhole\outin\ctx$ for every context
      $\ctx\neq\ctxhole$;
    \item 
      \emph{Contextual closure}: if $\ctx\outin\ctxtwo$ then
      $\ctxthreep\ctx\outin\ctxthreep\ctxtwo$ for any context
      $\ctxthree$.
    \end{enumerate}
    Note that $\outin$ can be seen as the prefix relation $\prefix$ on
      contexts. 
    The \emph{left-to-right order} $\ctx\leftright\ctxtwo$ is defined by
    \begin{enumerate}
    \item 
      \emph{Application}: if $\ctx\prefix \tm$ and
      $\ctxtwo\prefix\tmtwo$ then
      $\ctx\tmtwo\leftright\tm\ctxtwo$;
    \item 
      \emph{Substitution}: if $\ctx\prefix \tm$ and
      $\ctxtwo\prefix\tmtwo$ then
      $\ctx\esub\var\tmtwo\leftright\tm\esub\var\ctxtwo$;
    \item 
      \emph{Contextual closure}: if $\ctx\leftright\ctxtwo$ then
      $\ctxthreep\ctx\leftright\ctxthreep\ctxtwo$ for any context
      $\ctxthree$.
    \end{enumerate}      
    Last, the \emph{left-to-right outside-in order} is defined by
    $\ctx\leftout\ctxtwo$ if $\ctx\outin\ctxtwo$ or
    $\ctx\leftright\ctxtwo$.
  
\end{definition}

Two examples of the outside-in order are $(\l\var.\ctxhole)\tm
  \outin 
  (\l\var.(\ctxhole\esub\vartwo\tmtwo))\tm$ and $\tm\esub\var\ctxhole \outin \tm\esub\var{\tmtwo\ctx}$, and an example of the left-to-right order is ${\tm\esub\var\ctx}\tmtwo \leftright \tm\esub\var\tmthree\ctxhole$.
The next immediate lemma guarantees that we defined a total order.
\begin{lemma}[Totality of $\leftout$]\label{l:lefttor-basic} 
  If $\ctx\prefix\tm$ and $\ctxtwo\prefix\tm$ then either
  $\ctx\leftout\ctxtwo$ or $\ctxtwo\leftout\ctx$ or
  $\ctx=\ctxtwo$.
\end{lemma}


Remember that we identify redexes with their position context and write $\ctx\leftout\ctxtwo$. We can now define linear \lo\ reduction, first considered in \cite{DBLP:conf/popl/AccattoliBKL14}, where it is proved that it is standard and normalizing, and then in \cite{DBLP:conf/csl/AccattoliL14}, extending linear head reduction \cite{DBLP:journals/tcs/MascariP94,Danos04headlinear,DBLP:conf/rta/Accattoli12} to normal form.

\begin{definition}[Linear \lo Reduction $\tolo$]
\label{def:linear-lo-red}
  Let $\tm$ be a term. $\ctx$ is the
  \emph{leftmost-outermost} (\lo\ for short) redex of $\tm$ if $\ctx\leftout\ctxtwo$ for
  every other redex $\ctxtwo$ of $\tm$. We
  write $\tm\tolo\tmtwo$ if a step reduces
  the \lo\ redex.
\end{definition}


We now define \lo contexts and prove that the position of a linear \lo step is always a \lo context. We need two notions. 

\begin{definition}[\Quiet Term]
A term is \emph{\quiet} if it is $\to$-normal and it is not of the form $\sctxp{\la\var\tm}$.  
\end{definition}

\Quiet terms are such that their plugging in a context cannot create a multiplicative redex. We also need the notion of left free variable of a context, \ie\ of a variable occurring free at the left of the hole. 

\begin{definition}[Left Free Variables]
The set $\lfv\ctx$ of \emph{left free variables} of $\ctx$ is defined by:
\begin{align*}
    \lfv{\ctxhole} &\defeq \emptyset & \lfv{\tm\ctx} &\defeq \fv{\tm} \cup \lfv{\ctx} \\
    \lfv{\l\var.\ctx} &\defeq \lfv{\ctx} \setminus \set{\var} & \lfv{\ctx\esub{\var}{\tm}} &\defeq \lfv{\ctx} \setminus \set{\var} \\
    \lfv{\ctx\tm} & \defeq\lfv{\ctx} & \lfv{\tm\esub{\var}{\ctx}}                &\defeq (\fv{\tm} \setminus \set{\var})\cup \lfv{\ctx}
\end{align*}
\end{definition}


\begin{definition}[\lo\ Contexts]
  \label{def:LOCtx}
  A context $\ctx$ is 
  	\label{p:LOCtx-internal}\emph{\lo} if 
    \begin{enumerate}
     	\item \emph{Right Application}: whenever $\ctx = \ctxtwop{\tm\ctxthree}$ then $\tm$ is \quiet, and
	\item \emph{Left Application}: whenever $\ctx = \ctxtwop{\ctxthree\tm}$ then $\ctxthree\neq\sctxp{\la\var\ctxfour}$.
	\item \emph{Substitution}: whenever $\ctx = \ctxtwop{\ctxthree\esub\var\tmtwo}$ then $\var\notin\lfv\ctxthree$.
    \end{enumerate}
%
%
\end{definition}



\begin{lemma}[\lo Reduction and \lo\ Contexts]
	\label{l:LO-characts}
	Let $\tm \to \tmtwo$ by reducing a redex $\ctx$. Then $\ctx$ is a $\tolo$ step iff $\ctx$ is \lo.
\end{lemma}

\paragraph{Structural Equivalence.} A peculiar trait of the LSC is that the rewriting rules do not propagate ES. Therefore, evaluation is usually stable by structural equivalences moving ES around. In this paper we use the following equivalence, including garbage collection ($\tostructgc$), that we prove to be a strong bisimulation.
\begin{definition}[Structural equivalence]
The \emph{structural equivalence} $\eqstruct$ is the symmetric, reflexive, transitive, and contextual closure of
the following axioms:
$$
\begin{array}{rll@{\hspace{1em}}l}
    (\l\var.\tm)\esub{\vartwo}{\tmtwo}             & \tostructlam & \l\var.\tm\esub{\vartwo}{\tmtwo}               & \text{if $\var \not\in \fv{\tmtwo}$} \\
    (\tm\,\tmtwo)\esub{\var}{\tmthree}             & \tostructapl & \tm\esub{\var}{\tmthree}\,\tmtwo               & \text{if $\var \not\in \fv{\tmtwo}$} \\
    (\tm\,\tmtwo)\esub{\var}{\tmthree}             & \tostructapr & \tm\,\tmtwo\esub{\var}{\tmthree}               & \text{if $\var \not\in \fv{\tm}$} \\    
    \tm\esub{\var}{\tmtwo}\esub{\vartwo}{\tmthree} & \tostructcom & \tm\esub{\vartwo}{\tmthree}\esub{\var}{\tmtwo} & \text{if $\vartwo \not\in \fv{\tmtwo}$ and $\var \not\in \fv{\tmthree}$} \\
    \tm\esub{\var}{\tmtwo}\esub{\vartwo}{\tmthree} & \tostructes  & \tm\esub{\var}{\tmtwo\esub{\vartwo}{\tmthree}} & \text{if $\vartwo \not\in \fv{\tm}$} \\
    \tm\esub{\var}{\tmtwo}                         & \tostructgc  & \tm                                            & \text{if $\var \not\in \fv{\tm}$} \\
    \tm\esub{\var}{\tmtwo}                         & \tostructdup & \varsplit{\tm}{\var}{\vartwo}\esub{\var}{\tmtwo}\esub{\vartwo}{\tmtwo}   \\
\end{array}
$$
In $\tostructdup$, $\varsplit{\tm}{\var}{\vartwo}$ denotes a term obtained from $\tm$ by renaming some
(possibly none) occurrences of $\var$ as $\vartwo$, with $y$ a fresh variable.
\end{definition}

\begin{proposition}[Structural Equivalence $\eqstruct$ is a Strong Bisimulation]
	\label{prop:bisimulation}
	If $\tm \eqstruct\tmtwo \tolo \tmthree$ then exists $\tmfour$ s.t. $\tm \tolo\tmfour \eqstruct\tmthree$ and the steps are either both multiplicative or both exponential.
\end{proposition}
%

%% file: 03_-_Distilleries.tex
\section{Distilleries}
\label{sect:machines}
An abstract machine $\mach$ is meant to implement a strategy $\calculus$ via a \emph{distillation}, \ie\ a decoding function $\decodefun$. A machine has a state $\state$, given by a \emph{code} $\code$, \ie\ a $\l$-term $\tm$ without ES and not considered up to $\alpha$-equivalence, and some data-structures like stacks, dumps, environments, and heaps. The data-structures are used to implement the search for the next $\calculus$-redex and some form of substitution, and they decode to evaluation contexts for $\calculus$. Every state $\state$ decodes to a term $\decode\state$, having the shape $\stctx\state\ctxholep\code$, where $\code$ is the code currently under evaluation and $\stctx\state$ is the evaluation context given by the data-structures. 

A machine computes using transitions, whose union is denoted by $\tomach$, of two types. The \emph{principal} one, denoted by $\tomachp$, corresponds to the firing of a rule defining $\calculus$, up to structural equivalence $\eqstruct$. The \emph{commutative} transitions, denoted by $\tomachc$, only rearrange the data structures, and on the calculus are either invisible or mapped to $\eqstruct$. The terminology reflects a proof-theoretic view, as machine transitions can be seen as cut-elimination steps \cite{DBLP:journals/toplas/AriolaBS09,DBLP:conf/icfp/AccattoliBM14}. The transformation of evaluation contexts is formalized in the LSC as a structural equivalence $\eqstruct$, which is required to commute with evaluation $\calculus$, \ie\ to satisfy
\begin{center}
\begin{tabular}{c@{\sep}c@{\sep}c}
 \begin{tikzpicture}[ocenter]
  \node (s) {\normalsize$\tm$};
  \node at (s.center)  [below =0.7*\nodeVerDist](s2) {\normalsize$\tmtwo$};
  \node at (s.center) [right= 0.7*\nodeHorDist](t) {\normalsize$\tmfour$};

  \node at (s.center)[anchor = center, below=0.3*\nodeVerDist](eq1){\normalsize$\tostruct$};
    \draw[-o] (s) to  (t);
\end{tikzpicture} 

&  $\Rightarrow \exists \tmfive$ s.t. & 
 \begin{tikzpicture}[ocenter]
  \node (s) {\normalsize$\tm$};
  \node at (s.center)  [below =0.7*\nodeVerDist](s2) {\normalsize$\tmtwo$};
  \node at (s.center) [right= 0.7*\nodeHorDist](t) {\normalsize$\tmfour$};
  \node at (s2-|t) [](s1){\normalsize$\tmfive$};

  \node at (s.center)[anchor = center, below=0.3*\nodeVerDist](eq1){\normalsize$\tostruct$};
  \node at (t.center)[anchor = center, below=0.3*\nodeVerDist](eq2){\normalsize$\tostruct$};
    \draw[-o] (s) to  (t);
	\draw[-o, dashed] (s2) to  (s1);
\end{tikzpicture} 
\end{tabular}
\end{center}
for each of the rules of $\calculus$, preserving the kind of rule. In fact, this means that $\eqstruct$ is a \emph{strong} bisimulation (\ie\ \emph{one} step to \emph{one} step) with respect to $\calculus$, that is what we proved in \refprop{bisimulation} for the equivalence at work in this paper. Strong bisimulations formalize transformations which are transparent with respect to the behavior, even at the level of complexity, because they can be delayed without affecting the length of evaluation:

\begin{lemma}[Postponement of $\tostruct$]
	\label{l:postponement}
	If $\eqstruct$ is a strong bisimulation, $\tm\mathrel{(\calculus\cup\eqstruct)^*}\tmtwo$ implies $\tm\mathrel{\calculus^*\eqstruct}\tmtwo$ and the number and kind of steps of $\calculus$ in the two reduction sequences is exactly the same.
\end{lemma}

We can finally introduce distilleries, \ie\ systems where a strategy $\calculus$ simulates a machine $\mach$ up to structural equivalence $\eqstruct$ via the decoding $\decodefun$.

\begin{definition}
A \emph{distillery} $\distil = (\mach, \calculus, \tostruct,\decodefun)$ is given by:
\begin{enumerate}
	\item An \emph{abstract machine} $\mach$, given by
	\begin{enumerate}
		\item a deterministic labeled transition system (lts) $\tomach$ over states $\state$, with labels in $\{\mathtt{m},\mathtt{e},\mathtt{c}\}$; the transitions labelled by $\mathtt{m},\mathtt{e}$ are called \emph{principal}, the others \emph{commutative};
		\item a distinguished class of states deemed \emph{initial}, in bijection with closed \mbox{$\l$-terms}; from these, the \emph{reachable} states are obtained by applying $\tomach^\ast$;
	\end{enumerate}
	
	\item a deterministic \emph{strategy} $\calculus$, \ie, a deterministic lts over the terms of the LSC induced by some strategy on its reduction rules, with labels in $\{\mathtt{m},\mathtt{e}\}$. 

	 \item a \emph{structural equivalence} $\eqstruct$ on terms which is a strong bisimulation with respect to $\calculus$;

	\item a \emph{decoding function} $\decodefun$ from states to terms whose graph, when restricted to reachable states, is a weak simulation up to $\tostruct$ (the commutative transitions are considered as $\tau$ actions). More explicitly, for all reachable states:
	\begin{itemize}		
		\item \emph{projection of principal transitions}: $\state\tomachp\statetwo$ implies $\decode\state\calculus_{\mathtt p}\eqstruct\decode\statetwo$ for all $\mathtt p\in\{\mathtt m,\mathtt e\}$;
		\item \emph{distillation of commutative transitions}: $\state\tomacha\statetwo$ implies $\decode\state\eqstruct\decode\statetwo$.
	\end{itemize}
\end{enumerate}
\end{definition}

The simulation property is a minimum requirement, but a stronger form of relationship is usually desirable. Additional hypotheses are required in order to obtain the converse simulation and provide complexity bounds. 

\emph{Terminology:} an \emph{execution} $\exec$ is a sequence of transitions from an initial state. With $\size\exec$, $\sizep\exec$ and $\sizecom\exec$ we denote respectively the length, the number of principal and commutative transitions of $\exec$, whereas $\size\tm$ denotes the size of a term $t$.

\begin{definition}[Distillation Qualities]
A distillery is 
\begin{itemize}
	\item \emph{Reflective} when on reachable states:
\begin{itemize}
	\item \emph{Termination}: $\tomacha$ terminates;
	\item \emph{Progress}: if $\state$ is final then $\decode\state$ is a $\calculus$-normal form.
\end{itemize}

\item \emph{Bilinear} when,  given an execution $\exec$ from an initial term $\tm$:
\begin{itemize}
	\item \emph{Execution Length}: the number of commutative steps $\sizecom\exec$ is linear in both $\size\tm$ and $\sizep\exec$, \ie\ $\sizecom\exec\leq c\cdot(1+\sizep\exec)\cdot\size{\tm}$ for some non-zero constant $c$ (when $\sizep\exec=0$, $O(\size\tm)$ time is still needed to recognize that $t$ is normal).
	\item \emph{Commutative}: each commutative transition is implementable in $O(1)$ time on a RAM;
	\item \emph{Principal}: each principal transition is implementable in $O(\size\tm)$ time on a RAM.
\end{itemize}
\end{itemize}
\end{definition}

A reflective distillery is enough to obtain a weak bisimulation between the strategy $\togen$ and the machine $\mach$, up to structural equivalence $\eqstruct$ (again, the weakness is with respect to commutative transitions). With $\sizem\exec$ and  $\sizee\exec$ we denote respectively the number of multiplicative and exponential transitions of $\exec$.

\begin{theorem}[Correctness and Completeness]
	\label{tm:GenSim} 
	Let $\distil$ be a reflective distillery and $\state$ an initial state.
	\begin{enumerate}
		\item \emph{Simulation up to $\tostruct$}: for every execution $\exec:\state\tomach^*\statetwo$ there is a derivation $\deriv:\decode\state\togen^*\tostruct\decode\statetwo$ s.t.\ $\sizem\exec=\sizem\deriv$ and $\sizee\exec=\sizee\deriv$.
		\item \label{p:GenSim-two}\emph{Reverse Simulation up to $\tostruct$}: for every derivation $\deriv:\decode\state\togen^*\tm$ there is an execution $\exec:\state\tomach^*\statetwo$ s.t. $\tm\tostruct\decode\statetwo$ and $\sizem\exec=\sizem\deriv$ and $\sizee\exec=\sizee\deriv$.
	\end{enumerate}
\end{theorem}

Bilinearity, instead, is crucial for the low-level theorem.

\begin{theorem}[Low-Level Implementation Theorem]
\label{tm:low-level}
Let $\calculus$ be a strategy on terms with ES s.t. there exists a bilinear reflective distillery $\distil = (\mach, \calculus, \tostruct,\decodefun)$. Then a derivation $\deriv: \tm \togen^* \tmtwo$ is implementable on RAM machines in $O((1+\size\deriv)\cdot \size\tm)$ steps, \ie\ bilinear in the size $\size\tm$ of the initial term and the length $\size\deriv$ of the derivation.
\end{theorem}

\begin{proof}
Given $\deriv:\tm\calculus^n\tmtwo$ by \reftm{GenSim}.\refpointmute{GenSim-two} there is an execution $\exec:\state\tomach^*\statetwo$ s.t. $\tmtwo\tostruct\decode\statetwo$ and $\sizep\exec=\size\deriv$. The cost of implementing $\exec$ is the sum of the costs of implementing the commutative and the principal transitions. By bilinearity, $ \sizecom\exec =  O((1+\sizep\exec)\cdot\size{\tm})$ and so all the commutative transitions in $\exec$ require $O((1+\sizep\exec)\cdot\size{\tm})$ steps, because a single one takes a constant number of steps. Again by bilinearity, each principal one takes $O(\size{\tm})$, and so all the principal transitions together require $O(\sizep\exec\cdot\size{\tm})$ steps.\qed
\end{proof}

%% file: 04_-_Strengthening_the_MAM.tex
\section{Strengthening the MAM}

The machine we are about to introduce implements leftmost-outermost reduction and may therefore be seen as a strong version of the Krivine abstract machine (KAM). However, it differs from the KAM in the fundamental point of using global, as opposed to local, environments. It is therefore more appropriate to say that it is a strong version of the machine we introduced in \cite{DBLP:conf/icfp/AccattoliBM14}, which we called MAM (Milner abstract machine). Let us briefly recall its definition:
$$\small
  \setlength{\arraycolsep}{0.6em}
  \begin{array}{c|c|ccc|c|cl}
  \mbox{\scriptsize Code} & \mbox{\scriptsize Stack} & \mbox{\scriptsize Env} &&
  \mbox{\scriptsize Code} & \mbox{\scriptsize Stack} & \mbox{\scriptsize Env} \\
    \code\codetwo & \stack & \genv
    & \tomachasub{1} &
    \code & \codetwo\cons\stack & \genv\\
    
    \l\var.\code & \codetwo\cons\stack & \genv
    & \tomachm &
    \code & \stack & \esub{\var}{\codetwo}\cons\genv
  \\
    \var & \stack & \genv
    & \tomache &
    \rename{\code} & \stack & \genv & \text{ if $\genv(\var) = \code$}\\
    
  \end{array}
$$
Note that the stack and the environment of the MAM contain \emph{codes}, not \emph{closures} as in the KAM. A global environment indeed circumvents the complex mutually recursive notions of \emph{local environment} and \emph{closure}, at the price of the explicit $\alpha$-renaming $\rename{\code}$ which is applied \emph{on the fly} in $\tomache$. The price however is negligible, at least theoretically, as the asymptotic complexity of the machine is not affected, see \cite{DBLP:conf/icfp/AccattoliBM14} (the same can be said of variable names vs de Bruijn indexes/levels).

We know that the MAM performs \emph{weak} head reduction, whose reduction contexts are (informally) of the form $\ctxhole\stack$. This justifies the presence of the stack. It is immediate to extend the MAM so that it performs full head reduction, \ie, so that the head redex is reduced even if it is under an abstraction. Since head contexts are of the form $\Lambda.\ctxhole\stack$ (with $\Lambda$ a list of abstractions), we simply add a stack of abstractions $\Lambda$ and augment the machine with the following transition:
$$\small
  \setlength{\arraycolsep}{0.6em}
  \begin{array}{c|c|c|ccc|c|c|c}
  \mbox{\tiny Abs} & \mbox{\tiny Code} & \mbox{\tiny Stack} & \mbox{\tiny Env} &&
  \mbox{\tiny Abs} & \mbox{\tiny Code} & \mbox{\tiny Stack} & \mbox{\tiny Env}\\
    \Lambda & \l\var.\code & \stempty & \genv
    & \tomachasub{2} &
    \var\cons\Lambda & \code & \stempty & \genv
  \end{array}
$$
The other transitions do not touch the $\Lambda$ stack.

\lo reduction is nothing but iterated head reduction. \lo reduction contexts, which we formally introduced in Definition~\ref{def:LOCtx}, when restricted to the pure \mbox{$\l$-calculus} (without ES) are of the form $\Lambda.rC\pi$, where: $\Lambda$ and $\stack$ are as above; $r$, if present, is a \quiet term; and $C$ is either $\ctxhole$ or, inductively, a \lo context. Then \lo contexts may be represented by stacks of triples of the form $(\Lambda,r,\stack)$, where $r$ is a \quiet term. These stacks of triples will be called \emph{dumps}.

The states of the machine for full \lo reduction are as above but augmented with a dump and a \emph{phase} $\skphase$, indicating whether we are executing head reduction ($\skeval$) or whether we are backtracking to find the starting point of the next iteration ($\skback$). To the above transitions (which do not touch the dump and are always in the $\skeval$ phase), we add the following:

\scalemath{0.87}{
  \hspace{-0.8cm}
  \setlength{\arraycolsep}{0.6em}
  \footnotesize
  \begin{array}{c|c|c|c|c|ccc|c|c|c|c|c}
  \mbox{\tiny Abs} & \mbox{\tiny Code} & \mbox{\tiny Stack} & \mbox{\tiny Env} & \mbox{\tiny Dump} & \mbox{\tiny Ph} &&
  \mbox{\tiny Abs} & \mbox{\tiny Code} & \mbox{\tiny Stack} & \mbox{\tiny Env} & \mbox{\tiny Dump} & \mbox{\tiny Ph}\\
    \Lambda & \var & \stack & \genv & D & \skeval 
    & \tomachasub{3} &
    \Lambda & \var & \stack & \genv & D & \skback
      \\ \multicolumn{13}{r}{\text{ if $\genv(\var) = \bot$}}
  \\
     \var\cons\Lambda & \code & \stempty & \genv & D & \skback
    & \tomachasub{5} &
    \Lambda & \l\var.\code & \stempty & \genv & D & \skback
  \\
    \stempty & \codetwo & \stempty & \genv & (\Lambda,\code,\stack)\cons D & \skback
    & \tomachasub{7} &
    \Lambda & \code\codetwo & \stack & \genv & D & \skback
  \\
    \Lambda & \code & \codetwo\cons\stack & \genv & D & \skback
    & \tomachasub{6} &
    \stempty & \codetwo & \stempty & \genv & (\Lambda,\code,\stack)\cons D & \skeval 
  \\
  \end{array}
}

\noindent where $\genv(\var)=\bot$ means that the variable $\var$ is undefined in the environment $\genv$.

In the machine we actually use we join the dump and the $\Lambda$ stack into the \emph{frame} $\skframe$, to reduce the number of machine components (the analysis will however somewhat reintroduce the distinction). In the sequel, the reader should bear in mind that a state of the \skammachine introduced below corresponds to a  state of the machine just discussed according to the following correspondence:\footnote{Modulo the presence of markers of the form $\closescopem\var$ and $\openscopem\var$ in the environment, which are needed for bookkeeping purposes and were omitted here.}
\[\small\begin{array}{r@{\hspace{1.5em}}ccccccc}
   \mbox{Discussed Machine:}
  &

\begin{array}{c|c|c|c|c|c}
\mbox{\scriptsize Abs} & \mbox{\scriptsize Code} & \mbox{\scriptsize Stack} & \mbox{\scriptsize Env} & \mbox{\scriptsize Dump} & \mbox{\scriptsize Ph}\\
 \Lambda_0 & \code & \stack & \genv & (\Lambda_1,\code_1,\stack_1)\cons\cdots\cons(\Lambda_n,\code_n,\stack_n) & \skphase\\
 \end{array}\\
\multicolumn{2}{c}{\updownarrow}\\

\mbox{\skammachine:}& 
\begin{array}{c|c|c|c|c}
\mbox{\scriptsize Frame} & \mbox{\scriptsize Code} & \mbox{\scriptsize Stack} & \mbox{\scriptsize Env} & \mbox{\scriptsize Ph}\\
\Lambda_0\cons\skap{\code_1}{\stack_1}\cons\Lambda_1\cons\cdots\cons\skap{\code_n}{\stack_n}\cons\Lambda_n & \code & \stack & \genv & \skphase
 \end{array}\\

  \end{array}\]

%% file: 05_-_The_Strong_MAM.tex
\section{The Strong Milner Abstract Machine}
\label{sect:SMAM}
\input{strong-mam-figure}

The components and the transitions of the \skammachine\ are given by the first two boxes in Fig.~\ref{fig:SMAM}. As above, we use $\code,\codetwo,\ldots$ to denote \emph{codes}, \ie, terms not containing ES and \emph{well-named}, by which mean that distinct binders bind distinct variables and that the sets of free and bound variables are disjoint (codes are not considered up to $\alpha$-equivalence). The \skammachine has two phases: {\em evaluation} ($\skeval$) and {\em backtracking} ($\skback$).

\paragraph{Initial states.} The \emph{initial states} of the \skammachine\ are of the form $\skamstate{\skeval}{\stempty}{\code}{\stempty}{\stempty}$, where $\code$ is a closed code called the \emph{initial term}. In the sequel, we abusively say that a state is reachable from a term meaning that it is reachable from the corresponding initial state.

\paragraph{Scope Markers.} The two transitions to evaluate and backtrack on abstractions, $\tomachasubtwo$ and $\tomachasubfour$, add markers to delimit subenvironments associated to scopes. The marker $\openscopem\var$ is introduced when the machine starts evaluating under an abstraction $\lambda\var$, while $\closescopem\var$ marks the end of such a subenvironment. Note that the markers are not inspected by the machine. They are in fact needed only for the analysis, as they structure the frame and the environment of a reachable state into \emph{weak} and \emph{trunk} parts, allowing a simple decoding towards terms with ES.

\paragraph{Weak and Trunk Frames.}
A frame $\skframe$ may be uniquely decomposed into $\skframe=\skframeweak\cons\skframetrunk$ (with ``$\cons$'' abusively denoting concatenation, as we will always do in the sequel), where $\skframeweak=\skap{\code_1}{\stack_1}\cons\cdots\cons\skap{\code_n}{\stack_n}$ (with $n$ possibly null) is a \emph{weak frame}, \ie where no abstracted variable appear, and $\skframetrunk$ is a \emph{trunk frame}, \ie not of the form $(\code,\stack)\cons\skframe'$ (it either starts a variable entry or it is empty). More precisely, we rely on the alternative grammar\footnote{We slightly abuse notations: the production $\skframeweak\cons\skframetrunk$ may produce $\stempty\cons\stempty$ which is not a valid list/frame. To be formal, one should introduce the composition of lists, noted $\skframeweak\circ\skframetrunk$ or $\skframetrunk\ctxholep\skframeweak$ that removes empty frames in excess. To ease the reading, instead, we overload '$\cons$' with composition.} in the third box of \reffig{SMAM}. We denote by $\skl(\skframe)$ the set of variables in $\skframe$, \ie the set of $\var$ s.t. $\skframe= \skframetwo\cons\var\cons\skframethree$. 

\paragraph{Weak, Trunk, and Well-Formed Environments.}
Similarly to the frame, the environment of a reachable state has a weak/trunk structure. In contrast to frames, however, not every environment can be seen this way, but only the well-formed ones (reachable environments will be shown to be well-formed). A weak environment $\genvweak$ does not contain any open scope, \ie\ whenever in $\genvweak$ there is a scope opener marker ($\openscopem\var$) then one can also find the scope closer marker ($\closescopem\var$), and (globally) the closed scopes of $\genvweak$ are well-parenthesized. A trunk environment $\genvtrunk$ may instead also contain open scopes that have no closing marker in $\genvtrunk$ (but not unmatched closing markers $\closescopem\var$). Formally, weak $\genvweak$, trunk $\genvtrunk$, and well-formed environments $\genv$ (all  the environments that we will consider will be well-formed, that is why we note them $\genv$) are defined in the third box in \reffig{SMAM}. 
\paragraph{Accessing Environments and Meta-level Garbage Collection.} Fragments of the form $\closescopem\var\cons\genvweak\cons\openscopem\var$ within an environment will essentially be ignored; this is how a simple form of garbage collection is encapsulated at the meta-level in the decoding. In particular, for a well-formed environment $\genv$ we define $\genv(\var)$ as:
$$
        \begin{array}{rcl@{\hspace{3em}}rcl}
          \stempty(\var) & \defeq & \bot 
          &
          (\closescopem\vartwo\cons\genvweak\cons\openscopem\vartwo\cons\genv)(\var) & \defeq & \genv(\var) \\
          
          (\esub{\var}{\code}\cons\genv)(\var)       & \defeq & \code 
          &
          (\openscopem\var\cons\genv)(\var) & \defeq & \skboundvar\\
          
          (\esub{\vartwo}{\code}\cons\genv)(\var)       & \defeq & \genv(\var)
          &
          (\openscopem\vartwo\cons\genv)(\var) & \defeq & \genv(\var) \\
          
        \end{array}
$$
We write $\skl(\genv)$ to denote the set of variables bound to $\openscope$ by an environment $\genv$, \ie those variables whose
 scope is not closed with $\closescope$. 
  \begin{lemma}[Weak Environments Contain only Closed Scopes]
\label{l:weak-envs-closed-scopes}
If $\genvweak$ is a weak environment then $\sklp\genvweak = \emptyset$.
\end{lemma}
\paragraph{Implementation.} Variables are meant to be implemented as memory locations, so that the environment is simply a store, and accessing it takes constant time on RAM. In particular both the list structure of environments and the scope markers are used to define the decoding (\ie for the analysis), but are not meant to be part of the actual implementation. This is to kept in mind for the sake of the bilinearity of the distillery to be defined.

\paragraph{Compatibility.}
In the \skammachine, both the frame and the environment record information about the abstractions in which evaluation is currently taking place. Clearly, such information has to be coherent, otherwise the decoding of a state becomes impossible. The following compatibility predicate captures the correlation between the structure of the frame and that of the environment.

\begin{definition}[Compatibility $\skframe \compatible \genv$]
Compatibility $\skframe \compatible \genv$ between frames and environments is defined by
  \begin{enumerate}
  \item \emph{Base}: $\stempty \compatible \stempty$;
  \item \emph{Weak Extension}: $(\skframeweak\cons\skframetrunk) \compatible (\genvweak\cons\genvtrunk)$ if $\skframetrunk \compatible \genvtrunk$;
  \item \emph{Abstraction}: $(\var\cons\skframe) \compatible (\openscopem\var\cons\genv) $ if $\skframe \compatible \genv$;
  \end{enumerate}
\end{definition}


\begin{lemma}[Properties of Compatibility]
\label{l:comp-properties} 
\begin{enumerate}

\item \emph{Well-Formed Environments}: \label{p:comp-properties-well-form}
if $\skframe$ and $\genv$ are compatible then $\genv$ is well-formed.

\item \emph{Factorization}: \label{p:comp-properties-fact}
every compatible pair $\skframe\compatible \genv$ can be written as $ (\skframeweak\cons\skframetrunk) \compatible (\genvweak\cons\genvtrunk) $ with $\skframetrunk = \var\cons\skframetwo$ iff $\genvtrunk = \openscopem\var\cons\genvtwo$;

\item \emph{Open Scopes Match}: \label{p:comp-properties-open-scopes}
$\sklp\skframe = \sklp\genv$.

\item \emph{Compatibility and Weak Structures Commute}: \label{p:comp-properties-weak-comm}
for all $\skframeweak$ and $\genvweak$, $\skframe \compatible \genv$ iff $ (\skframeweak\cons\skframe) \compatible  (\genvweak\cons\genv)$.

\end{enumerate}
\end{lemma}

\paragraph{Invariants.} The properties of the machine that are needed to prove its correctness and completeness are given by the following invariants. 

\begin{lemma}[\skammachine invariants]
	\label{l:invariants}
	Let $\state = \skamstate{\skphase}{\skframe}{\codetwo}{\stack}{\genv}$ be a state reachable from an initial term $\code_0$. Then:
	\begin{enumerate}
		\item \emph{Compatibility:} \label{p:invariants-comp}
		$\skframe$ and $\genv$ are compatible, \ie $\skframe \compatible \genv$.

		\item \emph{Normal Form:} \label{p:invariants-nf}
		\begin{enumerate}[label=(\arabic*), ref=\arabic*]
			\item \emph{Backtracking Code}: \label{p:invariants-nf-code} 
			if $\skphase = \skback$, then $\codetwo$ is normal, and if $\stack$ is non-empty, then $\codetwo$ is \quiet; 
			
			\item \emph{Frame}: \label{p:invariants-nf-frame} 
			if $\skframe = \skframetwo\cons\skap{\codethree}{\stacktwo}\cons\skframethree$, then $\codethree$ is \quiet.
		\end{enumerate}

		\item \emph{Backtracking Free Variables:} \label{p:invariants-backtracking}
		\begin{enumerate}[label=(\arabic*), ref=\arabic*]
			\item \emph{Backtracking Code}:\label{p:invariants-backtracking-code} 
			if $\skphase = \skback$ then $\fv{\codetwo} \subseteq \skl(\skframe)$;

			\item \emph{Pairs in the Frame}: \label{p:invariants-backtracking-frame} 
			if $\skframe = \skframetwo\cons\skap{\codethree}{\stacktwo}\cons\skframethree$ then $\fv{\codethree} \subseteq \skl(\skframethree)$.
		\end{enumerate}

		\item \emph{Name:} \label{p:invariants-name}
		\begin{enumerate}[label=(\arabic*), ref=\arabic*]
			\item \emph{Substitutions}: \label{p:invariants-name-es}
			if $\genv = \genvtwo \cons \esub\var\code \cons \genvthree$  then $\var$ is fresh wrt $\code$ and $\genvthree$;

			\item \emph{Markers}: \label{p:invariants-name-marker}
			if $\genv = \genvtwo \cons \openscopem\var \cons \genvthree$ and $\skframe = \skframetwo \cons \var \cons \skframethree$ then $\var$ is fresh wrt $\genvthree$  and $\skframethree$, and $\genvtwo(\vartwo) = \bot$ for any free variable $\vartwo$ in $\skframethree$;

			\item \emph{Abstractions}: \label{p:invariants-name-abs}
			if $\la\var\code$ is a subterm of $\skframe$, $\codetwo$, $\stack$, or $\genv$ then $\var$ may occur only in $\code$ and in the closed subenvironment $\closescopem\var \cons \genvweak \cons \openscopem\var$ of $\genv$, if it exists.
		\end{enumerate}

		\item \emph{Closure:} \label{p:invariants-closure}
		\begin{enumerate}[label=(\arabic*), ref=\arabic*]
			\item \emph{Environment}: \label{p:invariants-closure-env}
			if $\genv = \genvtwo \cons \esub\var\code \cons \genvthree$ then $\genvthree(\vartwo)\neq \bot$ for all $\vartwo \in \fv\code$;
			
			\item \emph{Code, Stack, and Frame}: \label{p:invariants-closure-state}
			$\genv(\var)\neq \bot$ for any free variable in $\codetwo$ and in any code of $\stack$ and $\skframe$.
		\end{enumerate}
	\end{enumerate}
\end{lemma}

Since the statement of the invariants is rather technical, let us summarize the dependencies (or lack thereof) of the various points and their use in the distillation proof of the next section.
\begin{itemize}
	\item The compatibility, normal form and backtracking free variables invariants are independent of each other and of the subsequent invariants.
	\item The name invariant relies on the compatibility invariant only. It implies the determinism of the machine (because in the variable case at most one among $\tomache$ and $\tomachasubfour$ applies).
	\item The closure invariant relies on the compatibility, name and backtracking free variable invariants only. It is crucial for the progress property (because in the variable case at least one among $\tomache$ and $\tomachasubfour$ applies).
\end{itemize}
The proof of every invariant is by induction on the number of transitions leading to the reachable state. In this respect, the various points of the statement of each invariant are entangled, in the sense that each point needs to use the induction hypothesis of one of the other points, and thus they cannot be proved separately.

%% file: strong-mam-figure.tex
\begin{figure}[t]
	\begin{center}
	\ovalbox{$\begin{array}{l@{\hspace{1em}}rcl@{\hspace{2em}}l@{\hspace{1em}}rcl}
			\text{Frames} &
			\skframe & \grameq & \stempty
								\mid \skap{\code}{\stack}\cons\skframe
								\mid \var\cons\skframe
			&
			\text{Stacks} &
			\stack   & \grameq & \stempty \mid \code\cons\stack
			\\
			\text{Environments} &
			\genv    & \grameq & \stempty
								\mid \esub{\var}{\code}\cons\genv
								\mid \openscopem\var\cons\genv
								\mid \closescopem\var\cons\genv
			&
			\text{Phases} &
			\skphase & \grameq & \skeval \mid \skback
		\end{array}$}
	\end{center}
	\begin{center}
	  \ovalbox{
		\scalemath{0.94}{
		\setlength{\arraycolsep}{0.5em}
		\begin{array}{c|c|c|c|clc|c|c|c|c}
		\mbox{Frame} & \mbox{Code} & \mbox{Stack} & \mbox{Env} & \mbox{Ph}
		&&
		\mbox{Frame} & \mbox{Code} & \mbox{Stack} & \mbox{Env} & \mbox{Ph}\\
			\maca{ \skeval }{ \skframe }{ \code\codetwo }{ \stack }{ \genv } &
			\tomachasubone &
			\maca{ \skeval }{ \skframe }{ \code }{ \codetwo\cons\stack }{ \genv }
		\\
			\maca{ \skeval }{ \skframe }{ \l\var.\code }{ \codetwo\cons\stack }{ \genv } &
			\tomachm &
			\maca{ \skeval }{ \skframe }{ \code }{ \stack }{ \esub{\var}{\codetwo}\cons\genv }
		\\
			\maca{ \skeval }{ \skframe }{ \l\var.\code }{ \stempty }{ \genv } &
			\tomachasubtwo &
			\maca{ \skeval }{ \var\cons\skframe }{ \code }{ \stempty }{ \openscopem\var\cons\genv }
		\\
			\maca{ \skeval }{ \skframe }{ \var }{ \stack }{ \genv } &
			\tomache &
			\maca{ \skeval }{ \skframe }{ \rename{\code} }{ \stack }{ \genv } \\
			\multicolumn{11}{r}{\text{ if $\genv(\var) = \code$}}
		\\
			\maca{ \skeval }{ \skframe }{ \var }{ \stack }{ \genv } &
			\tomachasubthree &
			\maca{ \skback }{ \skframe }{ \var }{ \stack }{ \genv }
			\\ \multicolumn{11}{r}{\text{ if $\genv(\var) = \openscope$}}
		\\
			\maca{ \skback }{ \var\cons\skframe }{ \code }{ \stempty }{ \genv } &
			\tomachasubfour &
			\maca{ \skback }{ \skframe }{ \l\var.\code }{ \stempty }{ \closescopem\var\cons\genv }
		\\
			\maca{ \skback }{ \skap{\code}{\stack}\cons\skframe }{ \codetwo }{ \stempty }{ \genv } &
			\tomachasubfive &
			\maca{ \skback }{ \skframe }{ \code\codetwo }{ \stack }{ \genv }
		\\
			\maca{ \skback }{ \skframe }{ \code }{ \codetwo\cons\stack }{ \genv } &
			\tomachasubsix &
			\maca{ \skeval }{ \skap{\code}{\stack}\cons\skframe }{ \codetwo }{ \stempty }{ \genv }
		
		\end{array}
		}
		\vspace{-0.5cm}}
	\end{center}
	\begin{center}
    \ovalbox{
    \begin{tabular}{c|c}
      
      
      $\begin{array}[t]{lrll}
      \multicolumn{3}{c}{\scriptstyle \mbox{Frames (Ordinary, Weak, Trunk)}}\\
	\skframe & \grameq	&  \skframeweak \mid \skframetrunk \mid \skframeweak\cons\skframetrunk \\
	\skframeweak & \grameq	& \stempty \mid \skap{\code}{\stack}\cons\skframe\\
	\skframetrunk & \grameq	& \stempty \mid \var\cons\skframe 
      \end{array}$
      
      &

      $\begin{array}[t]{lrlll}
      \multicolumn{3}{c}{\scriptstyle \mbox{Environments (Well-Formed, Weak, Trunk)}}\\
	\genv		& \grameq & \genvweak \mid \genvtrunk \mid \genvweak \cons \genvtrunk\\
	\genvweak       & \grameq & \stempty \mid \esub{\var}{\code}\cons\genvweak \mid  
				    \closescopem\var\cons\genvweak\cons\openscopem\var\cons\genvweaktwo\\
	\genvtrunk	& \grameq & \stempty \mid \openscopem\var\cons\genv
      \end{array}$
      
    \end{tabular}    
    }
  \end{center}  
	\caption{The \skammachine.}
	\label{fig:SMAM}
	\vspace{-0.25cm}
\end{figure}

%% file: 06_-_Distilling_the_Strong_MAM.tex
\section{Distilling the Strong MAM}
\label{sect:distillation}

The definition of the decoding relies on the notion of compatible pair.

\begin{definition}[Decoding]
Let $\state = \mac{\skphase}{\skframe}{\code}{\stack}{\genv}$ be a state s.t. $\skframe \compatible \genv$ is a compatible pair. Then $\state$ decodes to a state context $\stctx\state$ and a term $\decode\state$ as follows:
\begin{center}
\ovalbox{\small
\begin{minipage}{0.97\textwidth}
$$\begin{array}{rcl@{\hspace{2em}}rcl}
    \multicolumn{3}{l}{\mbox{Weak Environments:}} & \multicolumn{3}{l}{\mbox{Compatible Pairs:}}\\
	\decode\stempty & \defeq & \ctxhole 
	&	
	\decodecompat{\stempty}{\stempty} & \defeq & \ctxhole 
	\\
	
    \decode{\esub{\var}{\codetwo}\cons\genvweak} & \defeq & \decodep\genvweak{\ctxhole\esub{\var}{\codetwo}} 
    &
    \decodecompat{ (\skframeweak\cons\skframetrunk) }{ (\genvweak\cons\genvtrunk) } 
    & \defeq & \decodecompat{\skframetrunk}{\genvtrunk}\ctxholep{\decodep\genvweak{\decode\skframeweak}}
    \\
    
    \decode{\closescopem\var\cons\genvweak\cons\openscopem\var\cons\genvweaktwo} & \defeq & \decode{\genvweaktwo}
    &
    \decodecompat{ (\var\cons\skframe) }{ (\openscopem\var\cons\genv) } & \defeq & \decodecompat{\skframe}{\genv}\ctxholep{\l\var.\ctxhole}
\end{array}$$
$$\begin{array}{rcl@{\hspace{5em}}rcl@{\hspace{3em}}rcl}
	\multicolumn{3}{l}{\mbox{Weak Frames:}} & \multicolumn{3}{l}{\mbox{Stacks:}} & \multicolumn{3}{l}{\mbox{States:}}\\
    \decode\stempty & \defeq & \ctxhole 
    &
    \decode\stempty & \defeq & \ctxhole
    &
    \stctx\state & \defeq & \decodecompat{\skframe}{\genv}\ctxholep{\decode\stack} 
    \\
    
    \decode{ \skap\codetwo\stack \cons\skframeweak} & \defeq & \decodep\skframeweak{\decodep\stack{\codetwo\ctxhole}}
    &
    \decode{\codetwo\cons\stack} & \defeq & \decode\stack\ctxholep{\ctxhole\codetwo}
    &
    \decode\state & \defeq & \stctxp\state\code\\
    \\
\end{array}$$
\end{minipage}
}
\end{center}
\end{definition}

The following lemmas sum up the properties of the decoding.

\begin{lemma}[Closed Scopes Disappear]
\label{l:closed-scopes-disap} 
Let $\skframe \compatible \genv$ be a compatible pair. Then $\decodecompat\skframe{ (\closescopem\var\cons\genvweak\cons\openscopem\var\cons\genv) } = \decodecompat\skframe\genv$.
\end{lemma}


\begin{lemma}[\lo Decoding Invariant]
\label{l:lo-decoding-invariant} 
Let $\state = \skamstate{\skphase}{\skframe}{\codetwo}{\stack}{\genv}$ be a
reachable state. Then $\decodecompat{\skframe}{\genv}$ and $\stctx\state$ are \lo contexts.
%
%
%
\end{lemma}


\begin{lemma}[Decoding and Structural Equivalence $\eqstruct$]
\label{l:eqstruct_structural_frames}
\begin{enumerate}
\item \label{p:eqstruct_structural_frames-stack} \sloppy{\emph{Stacks and Substitutions Commute}: if $\var$ does not occur free in $\stack$ then $\decodep\stack{\tm\esub\var\tmtwo} \eqstruct \decodep\stack\tm\esub\var\tmtwo$;}

\item \label{p:eqstruct_structural_frames-frame}\emph{Compatible Pairs Absorb Substitutions}: if $\var$ does not occur free in $\skframe$ then  $\decodecompat\skframe\genv\ctxholep{\tm\esub\var\tmtwo} \eqstruct \decodecompat\skframe{ (\esub\var\tmtwo\cons\genv) }\ctxholep\tm$.

\end{enumerate}
\end{lemma}

The next theorem is our first main result. By the abstract approach presented in \refsect{machines} (\reftm{GenSim}), it implies that the \skammachine is a correct and complete implementation of linear \lo evaluation to normal form.

\begin{theorem}[Distillation] 
	\label{thm:strong-mam-distillation}
	 $(\mbox{\skammachine}, \tolo, \eqstruct,\decodefun)$ is an explicit and reflective distillery. In particular:
 \begin{enumerate}
 	\item \emph{Projection of Principal Transitions}:
	\begin{enumerate}
		\item \emph{Multiplicative}: if $\state\tomachm\statetwo$ then $\decode\state\tom\eqstruct\decode\statetwo$;
		\item \emph{Exponential}: if $\state\tomache\statetwo$ then $\decode\state\toe\decode\statetwo$, duplicating the same subterm.
	\end{enumerate}
	
	 \item \emph{Distillation of Commutative Transitions}: 
 \begin{enumerate}
		\item \emph{Garbage Collection of Weak Environments}: if $\state\tomachcp{_{4}}\statetwo$ then $\decode\state \tostructgc \decode\statetwo$;

 	\item \emph{Equality Cases}: \label{p:strong-mam-distillation-most-comm} if $\state\tomachcp{_{1,2,3,5,6}}\statetwo$ then $\decode\state = \decode\statetwo$.

	\end{enumerate}
	\end{enumerate}
\end{theorem}
\begin{proof}
	\input{distillation_proof_2.tex}

\end{proof}

%% file: distillation_proof_2.tex
Recall, the decoding is defined as $\decode{\mac{\skphase}{\skframe}{\code}{\stack}{\genv}} \defeq \decodecompat{\skframe}{\genv}\ctxholep{\dstackp{\code}}$. Determinism of the machine follows by the name invariant (\reflemmap{invariants}{name}), and that of the strategy follows from the totality of the \lo order (\reflemma{lefttor-basic}). We list all cases but the simple equality ones, which may be found in \refapp{comm-distil}.
\begin{itemize}
\item \case{$\state = \mac{ \skeval }{ \skframe }{ \l\var.\code }{ \codetwo\cons\stack }{ \genv }
             \tomachm
             \mac{ \skeval }{ \skframe }{ \code }{ \stack }{ \esub{\var}{\codetwo}\cons\genv } = \statetwo$} 
             Note that $\stctx\statetwo = \decodecompat\skframe\genv \ctxholep{\decode\stack}$ is \lo by the \lo decoding invariant (\reflemma{lo-decoding-invariant}). Moreover by the closure invariant (\reflemmap{invariants}{closure}) $\var$ does not occur in $\skframe$ nor $\stack$, justifying the use of \reflemma{eqstruct_structural_frames} in:
     \[\begin{array}{lllllllll}    
    \decode{\mac{ \skeval }{ \skframe }{ \l\var.\code }{ \codetwo\cons\stack }{ \genv }}
    &=&
    \decodecompat\skframe{\genv}\ctxholep{\decode{\codetwo\cons\stack}\ctxholep{\l\var.\code}}\\
    &=&
    \decodecompat\skframe\genv\ctxholep{\dstackp{(\l\var.\code)\codetwo}}
    \\
    &\tom&
     \decodecompat\skframe{\genv}\ctxholep{\dstackp{\code\esub\var\codetwo}}\\

        & \eqstruct_{\reflemmaeqp{eqstruct_structural_frames}{stack}} &    
     \decodecompat\skframe{\genv}\ctxholep{\dstackp\code \esub\var\codetwo} \\
        & \eqstruct_{\reflemmaeqp{eqstruct_structural_frames}{frame}} &    
    \decodecompat\skframe{(\esub\var\codetwo\cons\genv)}\ctxholep{\dstackp\code } 
        & = &
    \decode{\mac{ \skeval }{ \skframe }{ \code }{ \stack }{ \esub{\var}{\codetwo}\cons\genv }}
    \end{array}\]

\item \case{$\state = \mac{ \skeval }{ \skframe }{ \var }{ \stack }{ \genv }
             \tomache
             \mac{ \skeval }{ \skframe }{ \rename{\code} }{ \stack }{ \genv } = \statetwo$ with $\genv(\var) = \code$} As before, $\stctx\state$ is \lo by \reflemma{lo-decoding-invariant}. Moreover,  $\genv(\var) = \code$ guarantees that $\genv$, and thus 
    $\stctx\state$, have a substitution
    binding $\var$ to $\code$. Finally, $\stctx\state = \stctx\statetwo$. Then
    \[\begin{array}{lllllllll}
    &
    \decode{\state}
    &=&
    \stctx\state\ctxholep{\var}
    &\toe&
    \stctx\state\ctxholep{\rename\code}
    &=&
    \decode{\statetwo}
    \end{array}\]

\item \case{$\state = \mac{ \skback }{ \var\cons\skframe }{ \code }{ \stempty }{ \genv }
             \tomachasubfour
             \mac{ \skback }{ \skframe }{ \l\var.\code }{ \stempty }{ \closescopem\var\cons\genv } = \statetwo$}
	By \reflemmap{invariants}{comp} $\var\cons\skframe \compatible \genv$, and by \reflemmap{comp-properties}{fact} $\genv = \genvweak\cons\openscopem\var\cons\genvtwo$. Then
	    \[\begin{array}{lllllllll}
   \decodecompat{(\var\cons\skframe)}{\genv}
     &=&
\decodecompat{(\var\cons\skframe)}{(\genvweak\cons\openscopem\var\cons\genvtwo)}
     &=&
\decodecompat{(\var\cons\skframe)}{(\openscopem\var\cons\genvtwo)}\ctxholep{\decode\genvweak}
     \end{array}\]

    Since we are in a backtracking phase ($\skback$), the backtracking free variables invariant (\reflemmapp{invariants}{backtracking}{code}) and the open scopes matching property (\reflemmap{comp-properties}{open-scopes}) give
    $\fv{\code} \subseteq_{\reflemmaeqp{invariants}{backtracking-code}} \sklp\skframe =_{\reflemmaeqp{comp-properties}{open-scopes}}  \sklp{\genvweak\cons\openscopem\var\cons\genvtwo} =_{\reflemmaeq{weak-envs-closed-scopes}} \sklp{\openscopem\var\cons\genvtwo}$, \ie $\decode{\genvweak}$ does not bind any variable in $\fv{\code}$.
Then $\decode{\genvweak}\ctxholep{\code} \tostructgc^* \code$, and
    \[\begin{array}{lllllllll}
     
     \decode{\mac{ \skback }{ \var\cons\skframe }{ \code }{ \stempty }{ \genv }}
     & = &
     \decodecompat{ (\var\cons\skframe) }{\genv}\ctxholep{\code}     
     \\
     & = &
     \decodecompat{ (\var\cons\skframe) }{(\genvweak\cons\openscopem\var\cons\genvtwo)}\ctxholep{\code}     
     \\
     & = &
     \decodecompat{ (\var\cons\skframe) }{ (\openscopem\var\cons\genvtwo) }\ctxholep{\decodep\genvweak\code}\\
     &\tostructgc^*&
     \decodecompat{ (\var\cons\skframe) }{ (\openscopem\var\cons\genvtwo)}\ctxholep\code\\

     & = &
     \decodecompat{ \skframe }{\genvtwo}\ctxholep{\la\var\code}\\
     & =_{\reflemmaeq{closed-scopes-disap}} &
     \decodecompat\skframe{(\closescopem\var\cons\genvweak\cons\openscopem\var\cons\genvtwo)}\ctxholep{\l\var.\code}\\
     & = &
     \decodecompat\skframe{ (\closescopem\var\cons\genv) }\ctxholep{\l\var.\code}
     
     &\!\!\!=&
     \!\decode{\mac{ \skback }{ \skframe }{ \l\var.\code }{ \stempty }{ \closescopem\var\cons\genv }}

     \end{array}\]
\end{itemize}

For what concerns reflectiveness, \emph{termination} of commutative transitions is subsumed by bilinearity (\reftm{bilinearity} below). For \emph{progress}, note that 
\begin{enumerate}
 \item \emph{the machine cannot get stuck during the evaluation phase}: for applications and abstractions it is evident and for variables one among $\tomache$ and $\tomachasubthree$ always applies, because of the closure invariant (\reflemmap{invariants}{closure}).
 \item \emph{final states have the form $\mac\skback\stempty\tm\stempty\genv$}, because 
 \begin{enumerate}
  \item by the previous consideration they are in a backtracking phase,
  \item if the stack is non-empty then $\tomachasubsix$ applies, 
  \item otherwise if the frame is not empty then either $\tomachasubfour$ or $\tomachasubfive$ applies.
 \end{enumerate}
 \item \emph{final states decode to normal terms}: a final state $\state = \mac\skback\stempty\tm\stempty\genv$ decodes to $\decode\state = \decodep\genv\tm$ which is normal and closed by the normal form (\reflemmapp{invariants}{nf}{code}) and backtracking free variables (\reflemmapp{invariants}{backtracking}{code}) invariants.\qed
\end{enumerate}

%% file: 07_-_Complexity_Analysis.tex
\section{Complexity Analysis}
\label{sect:complexity}
The complexity analysis requires a further invariant, bounding the size of the duplicated subterms. For us, $\codetwo$ is a subterm of $\code$ if it does so up to variable names, both free and bound. More precisely: define $\tm^-$ as $\tm$ in which all variables (including those appearing in binders) are replaced by a fixed symbol $\ast$. Then, we will consider $\tmtwo$ to be a subterm of $\tm$ whenever $\tmtwo^-$ is a subterm of $\tm^-$ in the usual sense. The key property ensured by this definition is that the size $\size\codetwo$ of $\codetwo$ is bounded by $\size\code$.

\begin{lemma}[Subterm Invariant]
\label{l:subterm-invariant} 
Let $\exec$ be an execution from an initial code $\code$.
Every code duplicated along $\exec$ using $\tomache$ is a subterm of $\code$.
 \end{lemma}

Via the distillation theorem (\refthm{strong-mam-distillation}), the invariant provides a new proof of the subterm property of linear \lo reduction (first proved in \cite{DBLP:conf/csl/AccattoliL14}).

\begin{lemma}[Subterm Property for $\tolo$]
\label{l:subterm-property-for-tolo}
 Let $\deriv$ be a $\tolo$-derivation from an initial term $\tm$. Every term duplicated along $\deriv$ using $\toe$ is a subterm of $\tm$.
\end{lemma}

The next theorem is our second main result, from which the low-level implementation theorem (\reftm{low-level}) follows. Let us stress that, despite the simplicity of the reasoning, the analysis is subtle as the length of backtracking phases (Point 2) can be bound only \emph{globally}, by the whole previous evaluation work.

\begin{theorem}[Bilinearity]
\label{tm:bilinearity}
The \skammachine is bilinear, \ie\ given an execution $\exec:\state\tomach^*\state'$ from an initial state of code 
$\tm$ then:
\begin{enumerate}
 \item \emph{Commutative Evaluation Steps are Bilinear}: $\sizecomev\exec \leq (1+\sizee\exec)\cdot\size\tm$.
 \item \emph{Commutative Evaluation Bounds Backtracking}: $\sizecombt\exec \leq 2\cdot\sizecomev\exec$.
 \item \emph{Commutative Steps are Bilinear}: $\sizecom\exec \leq 3\cdot(1+\sizee\exec)\cdot\size\tm$.
\end{enumerate}
\end{theorem}
\begin{proof}
	\input{bilinearity_proof}
\end{proof}

%% file: bilinearity_proof.tex
\begin{enumerate}
 \item We prove a slightly stronger statement, namely $\sizecomev\exec+\sizem\exec \leq (1+\sizee\exec)\cdot\size\tm$, by means of the following notion of size for stacks/frames/states:
\[\begin{array}{rcl@{\hspace{2em}}rcl}
\rmeasure{\stempty} & \defeq & 0 
&
\rmeasure{\var\cons\skframe} & \defeq & \rmeasure{\skframe}
\\

\rmeasure{\code\cons\stack} & \defeq & \size{\code} + \rmeasure{\stack} 
&
\rmeasure{\skap{\code}{\stack}\cons\skframe} & \defeq & \rmeasure{\stack} + \rmeasure{\skframe}\\

 \rmeasure{\mac{\skeval}{\skframe}{\code}{\stack}{\genv}} & \defeq & \rmeasure{\skframe} + \rmeasure{\stack} + \size{\code} 
 &
\rmeasure{\mac{\skback}{\skframe}{\code}{\stack}{\genv}} & \defeq & \rmeasure{\skframe} + \rmeasure{\stack}
   \end{array}\]
By direct inspection of the rules of the machine it can be checked that:
\begin{itemize}
\item \emph{Exponentials Increase the Size}: if $\state \tomache \statetwo$ is an exponential transition,
      then $\rmeasure{\statetwo} \leq \rmeasure{\state} + \size\tm$
      where $\size\tm$ is the size of the initial term;
      this is a consequence of the fact that exponential steps
      retrieve a piece of code from the environment, which is a subterm of the initial term by 
      \reflemma{subterm-invariant};
\item \emph{Non-Exponential Evaluation Transitions Decrease the Size}: if $\state \tomachhole{a} \statetwo$ with $a\in \set{\msym,\skeval\admsym_1,\skeval\admsym_2,\skeval\admsym_3}$ then
      $\rmeasure{\statetwo} < \rmeasure{\state}$;
\item \emph{Backtracking Transitions do not Change the Size}: if $\state \tomachhole{a} \statetwo$ with
      $a\in \set{\commfour,\commfive,\commsix}$ then $\rmeasure{\statetwo} = \rmeasure{\state}$.
\end{itemize}
Then a straightforward induction on $\size\exec$ shows that
\[\rmeasure{\statetwo} \leq \rmeasure{\state} + \sizee\exec \cdot \size\tm - \sizecomev\exec-\sizem\exec\]
\ie\ that $ \sizecomev\exec + \sizem\exec \leq \rmeasure{\state}  + \sizee\exec \cdot \size\tm - \rmeasure{\statetwo}$.

Now note that $\rmeasure{\cdot}$ is always non-negative and that since $\state$ is initial we have $\rmeasure{\state} = \size\tm$. We can then conclude with
\[\begin{array}{llllllllll}
   \sizecomev\exec + \sizem\exec & \leq &  \rmeasure{\state}  + \sizee\exec \cdot \size\tm - \rmeasure{\statetwo}\\
   & \leq & \rmeasure{\state}  + \sizee\exec \cdot \size\tm 
   & = & \size\tm  + \sizee\exec \cdot \size\tm & = & (1+\sizee\exec)\cdot\size\tm
  \end{array}\]

\item We have to estimate $\sizecombt\exec = \polsize\exec\commfour + \polsize\exec\commfive + \polsize\exec\commsix$. Note that
\begin{enumerate}
\item $\polsize\exec\commfour \leq \polsize\exec\commtwo$, as $\tomachasubfour$ pops variables from $\skframe$, pushed only by $\tomachasubtwo$;

\item $\polsize\exec\commfive \leq \polsize\exec\commsix$, as $\tomachasubfive$ pops pairs $\skap{\code}{\stack}$ from $\skframe$, pushed only by $\tomachasubsix$;

\item $\polsize\exec\commsix \leq \polsize\exec\commthree$, as $\tomachasubsix$ ends backtracking phases, started only by $\tomachasubthree$.
\end{enumerate}

Then $\sizecombt\exec \leq \polsize\exec\commtwo + 2\polsize\exec\commthree \leq 2\sizecomev\exec
$.

\item We have $\sizecom\exec = \sizecomev\exec + \sizecombt\exec \leq_{P. 2} \sizecomev\exec + 2\sizecomev\exec =_{P.1} 3\cdot(1+\sizee\exec)\cdot\size\tm $.
\end{enumerate}
Last, every transition but $\tomache$ takes a constant time on a RAM.
 The renaming in a $\tomache$ step is instead linear in $\size\code$, by the subterm invariant (\reflemma{subterm-invariant}). \qed

%% file: appendix.tex
\newpage
\appendix

\section{Proofs Omitted from \refsect{lo-reduction}\\ (Linear Leftmost-Outermost Reduction)}
The proofs omitted from \refsect{lo-reduction} are:
\begin{enumerate}
 \item \reflemma{lefttor-basic}, stating the totality of the $\leftout$ order. The proof is a trivial induction on $\tm$.
 \item \reflemma{LO-characts}, stating the equivalence of \lo contexts and \lo\ reduction. It is proved in the next subsection.
 \item \refprop{bisimulation}, stating that structural equivalence $\eqstruct$ is a strong bisimulation. The very long and tedious proof is postponed to the last section of the appendix, at page \pageref{app:bisimulation}.
\end{enumerate}

\subsection{Proof of the Equivalence of Definitions for \lo Contexts (\reflemma{LO-characts})}
\label{app:LO-characts}
\input{proofs/LO-contexts-equivalence-proof}

\section{Proofs Omitted from \refsect{machines}\\ (Distilleries)}
\label{app:machines}

The proof of \reflemma{postponement}, stating that a strong bisimulation $\eqstruct$ can be postponed, is a straightforward induction on the number of rewriting steps in $\tm\mathrel{(\calculus\cup\eqstruct)^*}\tmtwo$.\ms

The proof of \reftm{GenSim}, stating the correctness and completeness of the implementation for a reflective distillery, follows. The simulation is a simple proof by induction using the postponement lemma, while the reverse simulation is a similar induction following from the properties of a reflective distillery and by determinism of $\togen$.

\begin{proof}[of \reftm{GenSim}]\hfill
\begin{enumerate}
 \item \emph{Strong Simulation}: by induction on the length of $\exec$. If $\exec$ is empty then the empty derivation satisfies the statement. If $\exec$ is given by $\exectwo:\state\tomach^{k-1}\statethree$ followed by $\statethree\tomach\statetwo$. By \ih there exists $\derivtwo:\decode\state\togen^*\tostruct\decode\statethree$ s.t. $\sizep\exectwo=\size\derivtwo$. Cases of $\statethree\tomach\statetwo$:
 \begin{enumerate}
 \item \emph{Principal}: 
 by definition of a distillery, $\decode\statethree \togen\tostruct \decode\statetwo $, and so $\decode\state\togen^*\tostruct\decode\statethree \togen\tostruct \decode\statetwo$. By the postponement lemma (\reflemma{postponement}) the use of $\tostruct$ between $\togen^*$ and $\togen$ can be postponed, obtaining a term $\tmtwo$ and a derivation $\deriv$ s. t. $\deriv:\decode\state\togen^*\tmtwo\togen\tostruct \decode\statetwo$ with $\size\deriv = \size\derivtwo + 1 =_{\ih} \sizep\exectwo +1 = \sizep\exec$.
 \item \emph{Commutative}:
 by definition of a distillery, $\decode\statethree \tostruct \decode\statetwo $ , and so $\deriv:\decode\state\togen^*\tostruct\decode\statethree \tostruct \decode\statetwo$  verifies $\size\deriv = \size\derivtwo =_{\ih} \sizep\exectwo = \sizep\exec$. 
 \end{enumerate}
 
 \item \emph{Reverse Strong Simulation}: we use $\admnf\state$ 
 to denote the commutative normal form of $\state$, that exists and is unique because by hypothesis $\tomachc$ terminates and the machine is deterministic. The proof is by induction on the length of $\deriv$. If $\deriv$ is empty then the empty execution satisfies the statement.
 
If $\deriv$ is given by $\derivtwo:\decode\state\togen^*\tmtwo$ followed by $\tmtwo \togen \tm$ then by \ih there is an execution $\exectwo:\state\tomach^*\statethree$ s.t. $\tmtwo\tostruct\decode\statethree$ and $\sizep\exectwo = \size\derivtwo$. Note that since commutative transitions are distilled away, $\exectwo$ can be extended as
 $\exectwo':\state\tomach^*\statethree\tomachc^*\admnf{\statethree}$ with 
  $\tmtwo\tostruct\decode{\admnf{\statethree}}$ and $\sizep{\exectwo'}=\size\derivtwo$. Now, if $\tmtwo\togen\tm$ then 
  $\admnf{\statethree}$ cannot be a final state, otherwise there would be a contradiction with the progress hypothesis for a reflective distillery. 
  Then $\admnf{\statethree} \tomachp \statetwo$ 
  (the transition cannot be commutative because $\admnf{\statethree}$ is a commutative normal form). Now, by definition of distillery there exists $\tmthree$ s.t.\ $\decode{\admnf{\statethree}} \togen \tmthree \eqstruct \decode\statetwo$. But $\tmtwo\tostruct\decode{\admnf{\statethree}}\togen \tmthree$, so by \reflemma{postponement} there exists $\tm'$ s.t.\ $\tmtwo\togen\tm'\tostruct \tmthree\tostruct\decode\statetwo$. Now the determinism of $\togen$ implies $\tm'=\tm$, allowing us to conclude.
 \qed
 \end{enumerate}
\end{proof}

\section{Proofs Omitted from \refsect{SMAM}\\ (The Strong Milner Abstract Machine)}
\label{app:SMAM}

First of all, \reflemma{weak-envs-closed-scopes} (namely: \emph{If $\genvweak$ is a weak environment then $\sklp\genvweak = \emptyset$}) is proved by a straightforward induction on the definition of weak environment $\genvweak$.\ms

Then we prove the properties of compatibility (next subsection), and the invariants (\reflemma{invariants}). The proof of every invariant is studied separately, to stress the dependencies wrt to other invariants.

\subsection{Proof of the Properties of Compatibility (\reflemma{comp-properties})}
\label{app:props-of-comp}
\input{proofs/properties-of-compatibility-proof}

\subsection{Proof of the Compatibility Invariant (\reflemmap{invariants}{comp})}
\label{app:comp-invariant}
\input{proofs/compatibility-invariant}

\subsection{Proof of the Normal Form Invariant
(\reflemmap{invariants}{nf})}
\input{proofs/normal-form-invariant.tex}

\subsection{Proof of the Backtracking Free Variables Invariant (\reflemmap{invariants}{backtracking})}
\input{proofs/backtracking-free-vars-invariant.tex}

\subsection{Proof of the Name Invariant (\reflemmap{invariants}{name})}
\input{proofs/name-invariant.tex}

\subsection{Proof of the Closure Invariant (\reflemmap{invariants}{closure})}
\input{proofs/closed-environment-invariant-proof}

\section{Proofs Omitted from \refsect{distillation}\\ (Distilling the Strong MAM)}

\subsection{Proof of Closed Scopes Disappear (\reflemma{closed-scopes-disap})}
\input{proofs/decoding-and-compatibility-proof}

\subsection{Proof of the Leftmost-Outermost Invariant (\reflemma{lo-decoding-invariant})}
\label{app:LO-invariant}
\input{proofs/LO-invariant-proof.tex}

\subsection{Proof of the Properties of the Decoding wrt Structural Equivalence $\eqstruct$ (\reflemma{eqstruct_structural_frames})}
\input{proofs/decoding-and-struct-equiv-proof}

\subsection{Cases Omitted from the Proof of the Distillation Theorem (\refthm{strong-mam-distillation})}
\label{app:comm-distil}
\input{proofs/commutative-distallation-proof}

\section{Proofs Omitted from \refsect{complexity}\\ (Complexity Analysis)}

The proof of the subterm invariant (\reflemma{subterm-invariant}) for the machine is in the next subsection, and it is obtained as a corollary of a more general invariant. The subterm property for $\tolo$ (\reflemma{subterm-property-for-tolo}) is an immediate consequence of \reflemma{subterm-invariant} and the case of exponential transition in the distillation theorem (\refthm{strong-mam-distillation}).

\subsection{Proof of the Subterm Invariant (\reflemma{subterm-invariant})}
\label{app:subterm_invariant}
\input{proofs/subterm-invariant-proof}

\section{Proof that Structural equivalence is a Strong Bisimulation (\refprop{bisimulation})}
\label{app:bisimulation}
\input{bisimulation_proof.tex}

%% file: proofs/LO-contexts-equivalence-proof.tex
\begin{proof}
\hfill
		\begin{enumerate}
		\item[] $\Rightarrow$) There are three cases:
		\begin{enumerate}
			\item \emph{Left application}: if $\ctx = \ctxtwop{\ctxthree\tm}$ then clearly $\ctxthree\neq\sctxp{\la\var\ctxfour}$, otherwise $\ctx$ is not the position of the \lo\ redex.
			\item \emph{Right Application}: let $\ctx = \ctxtwop{\tmthree\ctxthree}$, and note $\tmthree$ is \quiet otherwise $\ctx$ is not the position of the \lo\ redex.
			\item \emph{Substitution}: if $\ctx = \ctxtwop{\ctxthree\esub\var\tmtwo}$ then $\var\notin\lfv\ctxthree$ otherwise there is an exponential redex of position $\leftout\ctx$, which would be absurd.
		\end{enumerate}
		
		\item[] $\Leftarrow$) Let $\ctxtwo$ the position of the $\tolo$ step in $\tm$ and suppose, for the sake of absurdity, that $\ctxtwo \neq \ctx$. By definition $\ctxtwo\leftout\ctx$. We have two cases:
		\begin{enumerate}
			\item $\ctxtwo\outin\ctx$. Then necessarily $\ctxtwo$ identifies a $\tom$-redex and we have $\ctx = \ctxtwop{\sctxp{\la\var\ctxthree} \tmthree}$. It follows that  $\ctx$ is not a \lo\ context, because the left application clause is contradicted, absurd.
			\item $\ctxtwo\leftright\ctx$. Then there is a decomposition $\ctx = \ctxthreep{\tmthree \ctxfour}$ with the hole of $\ctxtwo$ falling in $\tmthree$. By hypothesis $\tmthree$ is \quiet. Then $\tmthree = \ctx_0\ctxholep\var$ and the $\tolo$ step is a $\toe$-step substituting on $\var$ from a substitution in $\ctxthree$, \ie\ $\ctxthree = \ctx^\bullet\ctxholep{\ctx^\circ\esub\var\tm}$ for some contexts $\ctx^\bullet$ and $\ctx^\circ$. Then $\ctx =  \ctx^\bullet\ctxholep{\ctx^\circ\ctxholep{\tmthree \ctxfour}\esub\var\tm}$ and $\var\in\lfv{\ctx^\circ\ctxholep{\tmthree \ctxfour}}$, which contradicts the substitution clause in the hypothesis that $\ctx$ is a \lo\ context.\qed
		\end{enumerate}
	\end{enumerate}
\end{proof}

%% file: proofs/properties-of-compatibility-proof.tex

\begin{proof}
 The first three points (well-formed environments, factorization, open scopes) are by induction on the definition of compatible pair, and well-formed environments is omitted because it is evident. The fourth case is rather a corollary of factorization, and will be treated after the induction. The base case of the inductive reasoning is immediate for both factorization and open scopes. Two inductive cases:
 \begin{enumerate}
  \item \emph{Weak Extension}:
  \begin{enumerate}   
   \item \emph{Factorization}: the decomposition is immediate, and the correspondence about the first variable name follows from the \ih.
   \item \emph{Open Scopes}: by \ih, $\sklp\skframetrunk = \sklp\genvtrunk$. By \reflemma{weak-envs-closed-scopes}, $\sklp\genvweak = \emptyset$, and by definition $\sklp\skframeweak = \emptyset$. Then $\sklp\skframe = \sklp\skframeweak \cup \sklp\skframetrunk = \sklp\skframetrunk = \sklp\genvtrunk = \sklp\genvweak \cup \sklp\genvtrunk = \sklp\genv$.   
  \end{enumerate}

  \item \emph{Abstraction}
  \begin{enumerate}
   \item \emph{Factorization}: by definition $\var\cons\skframe$ and $\openscopem\var\cons\genv$ are a trunk frame $\skframetrunk$ and a trunk environment $\genvtrunk$, respectively.  given that $\cons$ is overloaded with composition, and weak trunk and environments can be empty we have $\skframetrunk = \empty\cons\skframetrunk$, and similarly for $\genvtrunk$, proving the decomposition property. The correspondence about the first variable name is evident.
   \item \emph{Open Scopes}: $\sklp{\var\cons\skframe} = \set\var \cup \sklp\skframe =_{\ih} \set\var \cup  \sklp\genv =  \sklp{\var\cons\genv}$.   
  \end{enumerate}
 \end{enumerate}
 
 \emph{Compatibility and Weak Structures Commute}: 
 \begin{enumerate}
  \item $\Rightarrow$) By factorization (\refpoint{comp-properties-fact}), $\skframe = \skframeweaktwo\cons\skframetrunk$ and $\genv = \genvweaktwo\cons\genvtrunk$. By definition of compatibility, if $\skframe \compatible \genv$ is derivable then $\skframetrunk \compatible \genvtrunk$ is also derivable. Now $\skframeweak\cons\skframeweaktwo$ and $\genvweak \cons\genvweaktwo$ are weak structures and so by the weak extension rule $\skframeweak\cons\skframe = \skframeweak\cons\skframeweaktwo\cons\skframetrunk \compatible \genvweak \cons \genvweaktwo\cons\genvtrunk = \genvweak \cons\genv$.
  \item $\Leftarrow$) By definition of compatibility, if $\skframeweak\cons\skframe = \skframeweak\cons\skframeweaktwo\cons\skframetrunk \compatible \genvweak \cons \genvweaktwo\cons\genvtrunk = \genvweak \cons\genv$ is derivable then $\skframetrunk \compatible \genvtrunk$ is also derivable, and $\skframe = \skframeweaktwo\cons\skframetrunk \compatible = \genvweaktwo\cons\genvtrunk = \genv$ by applying the weak extension rule.\qed
  \end{enumerate}
\end{proof}

%% file: proofs/compatibility-invariant.tex
\begin{proof}
By induction on the length of the number of transitions to reach $\state$. The invariant trivially holds for an initial state. For a non-empty evaluation sequence we list the cases for the last transitions. We only deal with those that act on the frame or on the environment, as the others immediately follows from the \ih.

  \begin{itemize}

  \item \case{$
          \mac{\skeval}{\skframe}{\l\var.\code}{\codetwo\cons\stack}{\genv}
          \tomachm
          \mac{\skeval}{\skframe}{\code}{\stack}{\esub\var\codetwo\cons\genv}
        $} 
        By \ih $\skframe$ and $\genv$ are compatible, \ie $\skframe = (\skframeweak\cons\skframetrunk) \compatible (\genvweak\cons\genvtrunk) = \genv$ with $\skframetrunk\compatible\genvtrunk$. Since $\esub\var\codetwo\cons\genvweak$ is still a weak environment, we have $(\skframeweak\cons\skframetrunk) \compatible (\esub\var\codetwo\cons\genvweak\cons\genvtrunk)$, \ie $\skframe \compatible (\esub\var\codetwo\cons\genv)$. 

  \item \case{$
          \mac{\skeval}{\skframe}{\l\var.\code}{\stempty}{\genv}
          \tomachasubtwo
          \mac{\skeval}{\var\cons\skframe}{\code}{\stempty}{\openscopem\var\cons\genv}
        $}
        By \ih $\skframe \compatible \genv$. By definition of compatibility we obtain 
      $(\var\cons\skframe) \compatible (\openscopem\var\cons\genv)$. 


  \item \case{$
          \mac{\skback}{\var\cons\skframe}{\code}{\stempty}{\genv}
          \tomachasubfour
          \mac{\skback}{\skframe}{\l\var.\code}{\stempty}{\closescopem\var\cons\genv}
        $}
        By \ih, $(\var\cons\skframe) \compatible \genv$. By the factorization property of compatible pairs (\reflemmap{comp-properties}{fact}) $\genv = \genvweak\cons\openscopem\var\cons\genvtwo$ with $\skframe \compatible \genvtwo$. Now $\closescopem\var\cons\genv = \closescopem\var\cons\genvweak\cons\openscopem\var\cons\genvtwo = \genvweaktwo\cons\genvtwo$. Then, from $\skframe \compatible \genvtwo$ by definition $\skframe \compatible (\genvweaktwo\cons\genvtwo)$, \ie  $\skframe \compatible (\closescopem\var\cons\genv)$. 

%

  \item \case{$
          \mac{\skback}{\skap{\code}{\stack}\cons\skframe}{\codetwo}{\stempty}{\genv}
          \tomachasubfive
          \mac{\skback}{\skframe}{\code\codetwo}{\stack}{\genv}
        $}
        By \ih, $(\skap{\code}{\stack}\cons\skframe) \compatible \genv$, so $\skframe \compatible \genv$ by \reflemmap{comp-properties}{weak-comm}. 

  \item \case{$
          \mac{\skback}{\skframe}{\code}{\codetwo\cons\stack}{\genv}
          \tomachasubsix
          \mac{\skeval}{\skap{\code}{\stack}\cons\skframe}{\codetwo}{\stempty}{\genv}
        $}
        By \ih, we have that $\skframe \compatible \genv$ which implies
      $(\skap{\code}{\stack}\cons\skframe) \compatible \genv$ by \reflemmap{comp-properties}{weak-comm}.
\qed
  \end{itemize}
 \end{proof}

%% file: proofs/normal-form-invariant.tex
\begin{proof}
  The invariant trivially holds for an initial state
  $\skamstate{\skeval}{\stempty}{\code}{\stempty}{\stempty}$.
For a non-empty evaluation sequence we list the cases for the last transitions. We only consider the cases for backtracking phases ($\skback$) or when the frame changes, the others ($\tomachasubone,\tomachm,\tomache$) are omitted because they follow immediately from the \ih.
  \begin{itemize}


  \item \case{$
          \mac{\skeval}{\skframe}{\l\var.\code}{\stempty}{\genv}
          \tomachasubtwo
          \mac{\skeval}{\var\cons\skframe}{\code}{\stempty}{\openscopem\var\cons\genv}
        $}
        \begin{enumerate}
      \item Trivial since $\skphase \neq \skback$.      
      \item Suppose $\var\cons\skframe$ can be written as $\var\cons\skframetwo\cons\skap{\codetwo}{\stacktwo}\cons\skframethree$.
            Then by \ih\ $\codetwo$ is a \quiet term.
      \end{enumerate}

  \item \case{$
          \mac{\skeval}{\skframe}{\var}{\stack}{\genv}
          \tomachasubthree
          \mac{\skback}{\skframe}{\var}{\stack}{\genv}
        $ with $\genv(\var) = \skboundvar$}
    Note that $\var \in \skl(\genv)$, because $\genv(\var) = \skboundvar$. 
    \begin{enumerate}
      \item $\var$ is a normal and \quiet term.      
      \item It follows from the \ih, as $\skframe$ is unchanged.
      \end{enumerate}

  \item \case{$
          \mac{\skback}{\var\cons\skframe}{\code}{\stempty}{\genv}
          \tomachasubfour
          \mac{\skback}{\skframe}{\l\var.\code}{\stempty}{\closescopem\var\cons\genv}
        $}
        	
      \begin{enumerate}
      \item By \ih\ we know that $\code$ is a normal form. Then $\l\var.\code$ is a normal form. the stack is empty, so we conclude.
      \item It follows from the \ih.
      \end{enumerate}

  \item \case{$
          \mac{\skback}{\skap{\code}{\stack}\cons\skframe}{\codetwo}{\stempty}{\genv}
          \tomachasubfive
          \mac{\skback}{\skframe}{\code\codetwo}{\stack}{\genv}
        $}
      \begin{enumerate}
      \item By \ih\ we have that $\codetwo$ is a normal term while by \refpoint{invariants-nf-frame} of the \ih $\code$ is \quiet. Therefore $\code\codetwo$ is a \quiet term.
      \item It follows from the \ih.
      \end{enumerate}

  \item \case{$
          \mac{\skback}{\skframe}{\code}{\codetwo\cons\stack}{\genv}
          \tomachasubsix
          \mac{\skeval}{\skap{\code}{\stack}\cons\skframe}{\codetwo}{\stempty}{\genv}
        $}
      \begin{enumerate}
      \item Trivial since $\skphase \neq \skback$.      
      \item $\code$ is a \quiet term by \refpoint{invariants-nf-code} of the \ih.\qed
      \end{enumerate} 
  \end{itemize}
\end{proof}

%% file: proofs/backtracking-free-vars-invariant.tex
\begin{proof}
  The invariant trivially holds for an initial state
  $\skamstate{\skeval}{\stempty}{\code_0}{\stempty}{\stempty}$
  if $\code_0$ is closed and well-named.
For a non-empty evaluation sequence we list the cases for the last transitions. We omit the transitions involving only states in evaluating phase, as for them everything follows immediately from the \ih.
  \begin{itemize}

  \item \case{$
          \mac{\skeval}{\skframe}{\vartwo}{\stack}{\genv}
          \tomachasubthree
          \mac{\skback}{\skframe}{\vartwo}{\stack}{\genv}
        $ with $\genv(\vartwo) = \openscope$}
     \begin{enumerate}

      \item \emph{Backtracking Code}: by hypothesis $\genv(\vartwo) = \openscope$, and so $\vartwo\in\sklp\genv =_{\reflemmaeqp{comp-properties}{open-scopes}} \sklp\skframe$.

      \item \emph{Pairs in the Frame}: it follows from the \ih.
    \end{enumerate}

  \item \case{$
          \mac{\skback}{\vartwo\cons\skframe}{\codethree}{\stempty}{\genv}
          \tomachasubfour
          \mac{\skback}{\skframe}{\l\vartwo.\codethree}{\stempty}{\closescopem\vartwo\cons\genv}
        $}

     \begin{enumerate}

      \item \emph{Backtracking Code}: by \ih $\fv\codethree\subseteq \sklp{\vartwo\cons\skframe}$ and so $\fv{\la\vartwo\codethree} = \fv{\codethree} \setminus \set\var = \sklp\skframe$.

      \item \emph{Pairs in the Frame}: it follows from the \ih.
    \end{enumerate}

  \item \case{$
          \mac{\skback}{\skap{\codethree}{\stack}\cons\skframe}{\codefour}{\stempty}{\genv}
          \tomachasubfive
          \mac{\skback}{\skframe}{\codethree\codefour}{\stack}{\genv}
        $}
     \begin{enumerate}
      \item \emph{Backtracking Code}: by \ih $\fv\codefour \subseteq \sklp{\skap{\codethree}{\stack}\cons\skframe} = \sklp{\skframe}$ and by \refpoint{invariants-backtracking-frame} of the \ih $\fv\codethree \subseteq \sklp{\skframe}$, and so $\fv{\codethree\codefour} \subseteq \sklp{\skframe}$.

      \item \emph{Pairs in the Frame}: it follows from the \ih.
    \end{enumerate}

  \item \case{$
          \mac{\skback}{\skframe}{\codethree}{\codefour\cons\stack}{\genv}
          \tomachasubsix
          \mac{\skeval}{\skap{\codethree}{\stack}\cons\skframe}{\codefour}{\stempty}{\genv}
        $}
 \begin{enumerate}
      \item \emph{Backtracking Code}: nothing to prove.

      \item \emph{Pairs in the Frame}: by \refpoint{invariants-backtracking-code} of the \ih $\fv\codethree \subseteq \sklp{\skframe}$, the rest follows from the \ih. \qed
    \end{enumerate}

  \end{itemize}

\end{proof}

%% file: proofs/name-invariant.tex
\begin{proof}
  The invariant trivially holds for an initial state
  $\skamstate{\skeval}{\stempty}{\codethree_0}{\stempty}{\stempty}$
  if $\codethree_0$ is closed and well-named.
For a non-empty evaluation sequence we list the cases for the last transitions:
  \begin{itemize}

  \item \case{$
          \mac{\skeval}{\skframe}{\codethree\codefour}{\stack}{\genv}
          \tomachasubone
          \mac{\skeval}{\skframe}{\codethree}{\codefour\cons\stack}{\genv}
        $}
		Every point follows from its \ih.

  \item \case{$
          \mac{\skeval}{\skframe}{\l\vartwo.\codethree}{\codefour\cons\stack}{\genv}
          \tomachm
          \mac{\skeval}{\skframe}{\codethree}{\stack}{\esub{\vartwo}{\codefour}\cons\genv}
        $}
    \begin{enumerate}
    \item \emph{Substitutions}: for $\esub{\vartwo}{\codefour}$ it follows from \refpoint{invariants-name-abs} of the \ih, for $\genv$ it follows from the \ih. 

    \item \emph{Markers}: note that by \refpoint{invariants-name-abs} of the \ih $\vartwo$ simply cannot occur in $\skframe$, the rest follows from the \ih.
    
    \item \emph{Abstractions}: it follows from the \ih.
    \end{enumerate} 

  \item \case{$
          \mac{\skeval}{\skframe}{\l\vartwo.\codethree}{\stempty}{\genv}
          \tomachasubtwo
          \mac{\skeval}{\vartwo\cons\skframe}{\codethree}{\stempty}{\openscopem\vartwo\cons\genv}
        $}

    \begin{enumerate}
    \item \emph{Substitutions}: it follows from the \ih.

    \item \emph{Markers}: for $\vartwo$ it follows from \refpoint{invariants-name-abs} of the \ih, the rest follows from the \ih.

    \item \emph{Abstractions}: it follows from the \ih.
    \end{enumerate} 

  \item \case{$
          \mac{\skeval}{\skframe}{\vartwo}{\stack}{\genv}
          \tomache
          \mac{\skeval}{\skframe}{\rename{\codethree}}{\stack}{\genv}
        $}
     It follows by the \ih and the fact that in $\rename{\codethree}$ the abstracted variables are renamed (wrt $\codethree$) with fresh names.

  \item \case{$
          \mac{\skeval}{\skframe}{\vartwo}{\stack}{\genv}
          \tomachasubthree
          \mac{\skback}{\skframe}{\vartwo}{\stack}{\genv}
        $}
         Every point follows from its \ih.

  \item \case{$
          \mac{\skback}{\vartwo\cons\skframe}{\codethree}{\stempty}{\genv}
          \tomachasubfour
          \mac{\skback}{\skframe}{\l\vartwo.\codethree}{\stempty}{\closescopem\vartwo\cons\genv}
        $}
By the compatibility invariant (\reflemmap{invariants}{comp}) $ (\vartwo\cons\skframe) \compatible \genv$, and by the factorization property of compatible pairs (\reflemmap{comp-properties}{fact}) $\genv = \genvweak\cons\openscopem\vartwo\cons\genvtwo$.

\begin{enumerate}
    \item \emph{Substitutions}: it follows from the \ih.

    \item \emph{Markers}: it follows from the \ih.

    \item \emph{Abstractions}: for $\la\vartwo\codethree$ it holds because by \refpoint{invariants-name-marker} of the \ih $\vartwo$ does not appear in $\skframe$ nor in $\genvtrunk$ (it may however occur in $\genvweak$, but this is taken into account by the statement). For the other abstractions \refpoint{invariants-name-marker} follows from the \ih.
    \end{enumerate}


  \item \case{$
          \mac{\skback}{\skap{\codethree}{\stack}\cons\skframe}{\codefour}{\stempty}{\genv}
          \tomachasubfive
          \mac{\skback}{\skframe}{\codethree\codefour}{\stack}{\genv}
        $}
             Every point follows from its \ih.

  \item \case{$
          \mac{\skback}{\skframe}{\codethree}{\codefour\cons\stack}{\genv}
          \tomachasubsix
          \mac{\skeval}{\skap{\codethree}{\stack}\cons\skframe}{\codefour}{\stempty}{\genv}
        $}
             Every point follows from its \ih.\qed

  \end{itemize}

\end{proof}

%% file: proofs/closed-environment-invariant-proof.tex
\begin{proof}
  The invariant trivially holds for an initial state
  $\skamstate{\skeval}{\stempty}{\code_0}{\stempty}{\stempty}$
  if $\code_0$ is closed and well-named.
For a non-empty evaluation sequence we list the cases for the last transitions:
  \begin{itemize}

  \item \case{$
          \mac{\skeval}{\skframe}{\codethree\codefour}{\stack}{\genv}
          \tomachasubone
          \mac{\skeval}{\skframe}{\codethree}{\codefour\cons\stack}{\genv}
        $}
		Every point follows from its \ih.

  \item \case{$
          \mac{\skeval}{\skframe}{\l\vartwo.\codethree}{\codefour\cons\stack}{\genv}
          \tomachm
          \mac{\skeval}{\skframe}{\codethree}{\stack}{\esub{\vartwo}{\codefour}\cons\genv}
        $}
        \begin{enumerate}
      \item \emph{Environment}: for $\esub{\vartwo}{\codefour}$ it follows from \refpoint{invariants-closure-state} of the \ih, for the rest it follows from the \ih. 
 
      \item \emph{Code, Stack, and Frame}: for $\vartwo$ is evident, as $\esub{\vartwo}{\codefour}\cons\genv$ is clearly defined on $\vartwo$, for the rest it follows from the \ih.  

    \end{enumerate}

  \item \case{$
          \mac{\skeval}{\skframe}{\l\vartwo.\codethree}{\stempty}{\genv}
          \tomachasubtwo
          \mac{\skeval}{\vartwo\cons\skframe}{\codethree}{\stempty}{\openscopem\vartwo\cons\genv}
        $}
    \begin{enumerate}
      \item \emph{Environment}: it follows from the \ih.
 
      \item \emph{Code, Stack, and Frame}: for $\vartwo$ is evident, as $\openscopem\vartwo\cons\genv$ is clearly defined on $\vartwo$, for the rest it follows from the \ih. 

    \end{enumerate}

  \item \case{$
          \mac{\skeval}{\skframe}{\vartwo}{\stack}{\genv}
          \tomache
          \mac{\skeval}{\skframe}{\rename{\codethree}}{\stack}{\genv}
        $}
    \begin{enumerate}
      \item \emph{Environment}: it follows from the \ih.
 
      \item \emph{Code, Stack, and Frame}: for $\rename{\codethree}$ it follows from \refpoint{invariants-closure-env} of the \ih, as $\codethree$ appears in the environment out of all closed scopes (otherwise the transition would not take place). The rest follows from the \ih.
      
    \end{enumerate}

  \item \case{$
          \mac{\skeval}{\skframe}{\vartwo}{\stack}{\genv}
          \tomachasubthree
          \mac{\skback}{\skframe}{\vartwo}{\stack}{\genv}
        $ with $\genv(\vartwo) = \openscope$}
     \begin{enumerate}
      \item \emph{Environment}: it follows from the \ih.
 
      \item \emph{Code, Stack, and Frame}: it follows from the \ih.
      
    \end{enumerate}

  \item \case{$
          \mac{\skback}{\vartwo\cons\skframe}{\codethree}{\stempty}{\genv}
          \tomachasubfour
          \mac{\skback}{\skframe}{\l\vartwo.\codethree}{\stempty}{\closescopem\vartwo\cons\genv}
        $}
By the compatibility invariant (\reflemmap{invariants}{comp}) $(\vartwo\cons\skframe) \compatible \genv$, and by the factorization property of compatible pairs (\reflemmap{comp-properties}{fact}) $\genv = \genvweak\cons\openscopem\vartwo\cons\genvtwo$.

     \begin{enumerate}
      \item \emph{Environment}: it follows from the \ih.
 
      \item \emph{Code, Stack, and Frame}: note that
      \begin{enumerate}
      \item $\genvweak$ does not bind any variable occurring free in $\codethree$ by \reflemmapp{invariants}{backtracking}{code},
      \item $\genvweak$ does not bind any variable occurring free in $\skframe$ by \reflemmapp{invariants}{name}{marker}, and
      \item the stack is empty by hypothesis.
      \end{enumerate}
      Then $\genvweak$ does not bind any free variable in the code, in the stack, nor in the frame, and we conclude using the \ih, because $\closescopem\var\genvweak\cons\openscopem\var\cons\genvtwo$ by definition is defined on a variable $\varthree$ iff $\genvtwo$ is.
      
    \end{enumerate}

  \item \case{$
          \mac{\skback}{\skap{\codethree}{\stack}\cons\skframe}{\codefour}{\stempty}{\genv}
          \tomachasubfive
          \mac{\skback}{\skframe}{\codethree\codefour}{\stack}{\genv}
        $}
     \begin{enumerate}
      \item \emph{Environment}: it follows from the \ih.
 
      \item \emph{Code, Stack, and Frame}: it follows from the \ih.
      
    \end{enumerate}

  \item \case{$
          \mac{\skback}{\skframe}{\codethree}{\codefour\cons\stack}{\genv}
          \tomachasubsix
          \mac{\skeval}{\skap{\codethree}{\stack}\cons\skframe}{\codefour}{\stempty}{\genv}
        $}
 \begin{enumerate}
      \item \emph{Environment}: it follows from the \ih.
 
      \item \emph{Code, Stack, and Frame}: it follows from the \ih.\qed
      
    \end{enumerate}

  \end{itemize}

\end{proof}

%% file: proofs/decoding-and-compatibility-proof.tex
\begin{proof}
Essentially it follows from $\decode{\closescopem\var\cons\genvweak\cons\openscopem\var\cons\genv} = \decode\genv$. Precisely, by \reflemmap{comp-properties}{fact} $\skframe$ and $\genv$ have, respectively, the forms $\skframeweak\cons\skframetrunk$ and $\genvweaktwo\cons\genvtrunk$. Now, 
  
  \[\begin{array}{cclcccccc}
    \decodecompat\skframe{ (\closescopem\var\cons\genvweak\cons\openscopem\var\cons\genv) } & = &
    \decodecompat{ (\skframeweak\cons\skframetrunk) }{ (\closescopem\var\cons\genvweak\cons\openscopem\var\cons\genvweaktwo\cons\genvtrunk) }
    \\
     & = & 
    \decodecompat\skframetrunk\genvtrunk \ctxholep{\decodep{\closescopem\var\cons\genvweak\cons\openscopem\var\cons\genvweaktwo}{\decode\skframeweak}}
    \\
    
     & = & 
    \decodecompat\skframetrunk\genvtrunk \ctxholep{\decodep\genvweaktwo{\decode\skframeweak}}
    \\
    
    & = &
    \decodecompat{ (\skframeweak\cons\skframetrunk) }{ (\genvweaktwo\cons\genvtrunk) }
    & = &
    \decodecompat\skframe\genv   
  \end{array}\]\qed
\end{proof}

%% file: proofs/LO-invariant-proof.tex

For the invariant we need the following lemma.

\begin{lemma}[Compatible Pairs Decode to Non-Applicative Contexts]
\label{l:decodecompat-not-ap} 
Let $\skframeweak$ be a weak frame, $\genvweak$ a weak environment, and $\skframe \compatible \genv$ a compatible pair. 
Then $\decode\skframeweak$, $\decode\genvweak$, and $\decodecompat{\skframe}{\genv}$ are contexts that are not applicative, \ie not of the form $\ctxp{\sctx\tm}$. 
\end{lemma}

\begin{proof} The fact that $\decode\skframeweak$ and $\decode\genvweak$ are not applicative is an immediate induction over their structure. For $\decodecompat\skframe\genv$ we reason by induction on the compatibility of $\skframe$ and $\genv$. The base case $\decodecompat\stempty\stempty = \ctxhole$ is evident. Inductive cases:

  \begin{enumerate}
    \item \emph{Weak Extension}, \ie $(\skframeweak\cons\skframetrunk) \compatible (\genvweak\cons\genvtrunk) $ with $\skframetrunk \compatible \genvtrunk$. By \ih $\decodecompat\skframetrunk\genvtrunk$ is not applicative and both $\decode\skframeweak$ and $\decode\genvweak$ are not applicative. By definition, $\decodecompat{ (\skframeweak\cons\skframetrunk) }{ (\genvweak\cons\genvtrunk) } = \decodecompat{\skframetrunk}{\genvtrunk}\ctxholep{\decodep\genvweak{\decode\skframeweak}}$, which is then not applicative.
  
    \item \emph{Abstraction}, \ie $ (\var\cons\skframe) \compatible (\openscopem\var\cons\genv) $ with $\skframe \compatible \genv$. Immediate, as $\decodecompat{\skframe}{\genv}\ctxholep{\l\var.\ctxhole}$ is not applicative.\qed
  \end{enumerate}
\end{proof}

We can now prove that the decoding of the data-structures of a reachable state is a \lo context.

\begin{proof}[Leftmost-Outermost Invariant, \reflemma{lo-decoding-invariant}]

We prove that $\decodecompat{\skframe}{\genv}$ is a \lo context, the fact that $\stctx\state$ is a \lo contexts then easily follows, as $\stctx\state \defeq \decodecompat{\skframe}{\genv}\ctxholep{\decode\stack}$.

The invariant trivially holds for an initial state
  $\skamstate{\skeval}{\stempty}{\code_0}{\stempty}{\stempty}$.
For a non-empty evaluation sequence we list the cases for the last transitions. We omit the cases for which the environment and the frame do not change (\ie $ \tomachasubone,\tomache,\tomachasubthree$), as for them the statement follows from the \ih.
  \begin{itemize}

  \item \case{$
          \mac{\skeval}{\skframe}{\l\var.\code}{\codetwo\cons\stack}{\genv}
          \tomachm
          \mac{\skeval}{\skframe}{\code}{\stack}{\esub{\var}{\codetwo}\cons\genv}
        $}
        By \ih\ $\decodecompat{\skframe}{\genv}$ is \lo. Let $\skframe = \skframeweak\cons\skframetrunk$, so that $\decodecompat{\skframe}{\genv}=\decodecompat{\skframetrunk}{\genv}\ctxholep{\decode{\skframeweak}}$.
      Note that, by the name invariant (\reflemmapp{invariants}{name}{abs}), the eventual occurrences of $\var$ are all in $\code$ and so $\var \not\in \fv{\decode{\skframeweak}}$, and in particular $\var \not\in \lfv{\decode{\skframeweak}}$. Then, $\decodecompat{\skframetrunk}{\genv}\ctxholep{\decode{\skframeweak}\esub{\var}{\codetwo}}$ is \lo: the conditions of \refdefp{LOCtx}{internal} are satisfied either because $\decodecompat{\skframe}{\genv}=\decodecompat{\skframetrunk}{\genv}\ctxholep{\decode{\skframeweak}}$ is \lo or because $\var \not\in \lfv{\decode{\skframeweak}}$.

  \item \case{$
          \mac{\skeval}{\skframe}{\l\var.\code}{\stempty}{\genv}
          \tomachasubtwo
          \mac{\skeval}{\var\cons\skframe}{\code}{\stempty}{\openscopem\var\cons\genv}
        $}
        By \ih\ we have $\decodecompat{\skframe}{\genv}$ is \lo and by \reflemma{decodecompat-not-ap} $\decodecompat{\skframe}{\genv}$ is not applicative, so
      $\decodecompat{ (\var\cons\skframe) }{ (\openscopem\var\cons\genv) } = \decodecompat{\skframe}{\genv}\ctxholep{\l\var.\ctxhole}$
      is \lo (it satisfies the conditions of \refdefp{LOCtx}{internal} because $\decodecompat{\skframe}{\genv}$ does).

  \item \case{$
          \mac{\skback}{\var\cons\skframe}{\code}{\stempty}{\genv}
          \tomachasubfour
          \mac{\skback}{\skframe}{\l\var.\code}{\stempty}{\closescopem\var\cons\genv}
        $}
        	By the compatibility invariant (\reflemmap{invariants}{comp}) $ (\var\cons\skframe) \compatible \genv$, and by the factorization property of compatible pairs (\reflemmap{comp-properties}{fact}) $\genv = \genvweak\cons\openscopem\var\cons\genvtwo$. By definition
      $$
          \decodecompat{ (\var\cons\skframe) }{ (\genvweak\cons\openscopem\var\cons\genvtrunk) }
        = \decodecompat{\skframe}{\genvtrunk}\ctxholep{\l\var.\decode{\genvweak}}
      $$
      that by \ih is \lo. Now, $\decodecompat{\skframe}{\genvtrunk}$ is \lo, as it satisfies the conditions of \refdefp{LOCtx}{internal} because $\decodecompat{\skframe}{\genv}$ does. We conclude by noticing that the compatible pair of the target state satisfies $\decodecompat{\skframe}{ (\closescopem\var\cons\genv) } =
      \decodecompat{\skframe}{(\closescopem\var\cons\genvweak\cons\openscopem\var\cons\genvtrunk)} =_{\reflemmaeq{closed-scopes-disap}} \decodecompat{\skframe}{\genvtrunk}$.

  \item \case{$
          \mac{\skback}{\skap{\code}{\stack}\cons\skframe}{\codetwo}{\stempty}{\genv}
          \tomachasubfive
          \mac{\skback}{\skframe}{\code\codetwo}{\stack}{\genv}
        $}
        By \ih\ we have that
      $\decodecompat{ (\skap{\code}{\stack}\cons\skframe) }{\genv}$ is \lo and by \emph{frame} part of the backtracking normal form invariant (\reflemmapp{invariants}{backtracking}{frame}) $\code$ is \quiet. By definition, 
      $\decodecompat{ (\skap{\code}{\stack}\cons\skframe) }{\genv} = \decodecompat{\skframe}{\genv}\ctxholep{\dstackp{\code\ctxhole}}$,
      Then, $\decodecompat{\skframe}{\genv}$---being a prefix of $\decodecompat{ (\skap{\code}{\stack}\cons\skframe) }{\genv}$---verifies the conditions of \refdefp{LOCtx}{internal} and is \lo.

  \item \case{$
          \mac{\skback}{\skframe}{\code}{\codetwo\cons\stack}{\genv}
          \tomachasubsix
          \mac{\skeval}{\skap{\code}{\stack}\cons\skframe}{\codetwo}{\stempty}{\genv}
        $}
        Note that 
      \begin{enumerate}
       \item  $\decodecompat{\skframe}{\genv}$ is \lo by \ih,
       \item  $\decodecompat{\skframe}{\genv}$ is not applicative by \reflemma{decodecompat-not-ap},
       \item \label{p:three} $\fv\code \subseteq \skl(\skframe)$ by the backtracking free variables invariant (\reflemmapp{invariants}{backtracking}{code}).
       \item $\code$ is a \quiet term by the normal form invariant (\reflemmapp{invariants}{nf}{code}), because the stack at the left-hand side is not empty.
      \end{enumerate}
      Note that \refpoint{three} guarantees that $\var\notin\fv\code$, and so in particular $\var\notin\lfv\code$, for any ES $\esub\var\codethree$ in $\genv$ (and so in $\decodecompat{\skframe}{\genv}$). Then $\decodecompat{\skframe}{\genv}\ctxholep{\dstackp{\code\ctxhole}}$ is \lo (because it verifies the conditions of \refdefp{LOCtx}{internal}, by the listed points),
      that is to say $\decodecompat{ (\skap{\code}{\stack}\cons\skframe) }{\genv}$ is \lo.\qed

  \end{itemize}
\end{proof}

%% file: proofs/decoding-and-struct-equiv-proof.tex

We here present a more general statement than the one in the paper. The reason is that the proof of the second point of the lemma (\emph{Compatible Pairs Absorb Substitutions}) actually requires a further lemma (\emph{Weak Frames and Substitutions Commute} below) that is omitted from the statement in the paper because it is not used anywhere else.

\begin{lemma}[Decoding and Structural Equivalence $\eqstruct$]
\label{l:eqstruct_structural_frames-app}
\begin{enumerate}
\item \label{p:eqstruct_structural_frames-app-stack} \emph{Stacks and Substitutions Commute}: if $\var$ does not occur free in $\stack$ then $\decodep\stack{\tm\esub\var\tmtwo} \eqstruct \decodep\stack\tm\esub\var\tmtwo$;

\item \label{p:eqstruct_structural_frames-app-frame-weak} \emph{Weak Frames and Substitutions Commute}: if $\var$ does not occur free in $\skframeweak$ then $\decodep\skframeweak{\tm\esub\var\tmtwo} \eqstruct \decodep\skframeweak\tm\esub\var\tmtwo$;

\item \label{p:eqstruct_structural_frames-app-pairs}\emph{Compatible Pairs Absorb Substitutions}: if $\var$ does not occur free in $\skframe$ then  $\decodecompat\skframe\genv\ctxholep{\tm\esub\var\tmtwo} \eqstruct \decodecompat\skframe{ (\esub\var\tmtwo\cons\genv) }\ctxholep\tm$.

\end{enumerate}
\end{lemma}

\begin{proof}
\hfill
 \begin{enumerate}
  \item \emph{Stacks and Substitutions Commute}: by induction on $\stack$. Cases:
    \begin{enumerate}
     \item \emph{Empty Stack}, \ie $\stack = \stempty$. Then $\decodep\stempty{\tm\esub\var\tmtwo} = \tm\esub\var\tmtwo = \decodep\stempty\tm\esub\var\tmtwo$.
     \item \emph{Non-Empty Stack}, \ie $\stack = \codethree\cons\stacktwo$. Then 
     \[\begin{array}{llllllllll}
     \decodep{\codethree\cons\stacktwo}{\tm\esub\var\tmtwo} 
     & = & 
     \decodep{\stacktwo}{\tm\esub\var\tmtwo}\codethree \\
     
     & \eqstruct_{\ih} & 
     \decodep{\stacktwo}{\tm}\esub\var\tmtwo\codethree 
     \\
     
     & \tostructapr & \decodep{\stacktwo}{\tm}\codethree\esub\var\tmtwo 
     & = & 
     \decodep{\codethree\cons\stacktwo}\tm\esub\var\tmtwo
      \end{array}\]
    \end{enumerate}
   Note that the proof uses only $\tostructapl$.

   \item \emph{Weak Frames and Substitutions Commute}: by induction on $\skframeweak$. Cases:
    \begin{enumerate}
     \item \emph{Empty Weak Frame}, \ie $\skframeweak = \stempty$. Then $\decodep\stempty{\tm\esub\var\tmtwo} = \tm\esub\var\tmtwo = \decodep\stempty\tm\esub\var\tmtwo$.
     \item \emph{Non-Empty Weak Frame}, \ie $\skframeweak = \skap\codethree\stack \cons\skframeweaktwo$. Then 
     \[\begin{array}{llllllllll}
     \decodep{\skap\codethree\stack \cons\skframeweaktwo}{\tm\esub\var\tmtwo} 
     & = & 
     \decodep\skframeweaktwo{\decodep\stack{\codethree (\tm\esub\var\tmtwo)}}\\
     
     & \tostructapr & 
     \decodep\skframeweaktwo{\decodep\stack{(\codethree \tm)\esub\var\tmtwo}}\\
     
     & \tostruct_{\refpointeq{eqstruct_structural_frames-app-stack}} & 
     \decodep\skframeweaktwo{\decodep\stack{(\codethree \tm)}\esub\var\tmtwo}\\
     
     & \eqstruct_{\ih} & 
     \decodep\skframeweaktwo{\decodep\stack{\codethree \tm}}\esub\var\tmtwo
     & = & 
     \decodep{\skap\codethree\stack \cons\skframeweaktwo}{\tm} \esub\var\tmtwo
      \end{array}\]
    \end{enumerate}
   Note that the proof uses only $\tostructapr$ and $\tostructapl$ (because of the previous point). 
   
   \item \emph{Compatible Pairs Absorb Substitutions}: By \reflemmap{comp-properties}{fact} we can decompose $\skframe$ and $\genv$ in their weak and trunk parts, obtaining:
    \[\begin{array}{llllllllll}
      \decodecompat\skframe\genv\ctxholep{\tm\esub\var\tmtwo}
      & = & 
      \decodecompat{ (\skframeweak\cons\skframetrunk) }{ (\genvweak\cons\genvtrunk}\ctxholep{\tm\esub\var\tmtwo) }\\
      & = & 
      \decodecompat{\skframetrunk}{\genvtrunk}\ctxholep{\decodep\genvweak{\decodep\skframeweak{\tm\esub\var\tmtwo}}}\\
     
      & =_{\refpointeq{eqstruct_structural_frames-app-frame-weak}} &
      \decodecompat{\skframetrunk}{\genvtrunk}\ctxholep{\decodep\genvweak{\decodep\skframeweak{\tm}\esub\var\tmtwo}}\\
      & = & 
      \decodecompat{\skframetrunk}{\genvtrunk}\ctxholep{\decodep{\esub\var\tmtwo\cons\genvweak}{\decodep\skframeweak{\tm}}}\\
      & = & 
      \decodecompat{( \skframeweak\cons\skframetrunk) }{( \esub\var\tmtwo\cons\genvweak\cons\genvtrunk) }\ctxholep{\tm}   
      & = & 
      \decodecompat\skframe{ (\esub\var\tmtwo\cons\genv) }\ctxholep\tm
      \end{array}\]

  \end{enumerate}  \qed
\end{proof}

%% file: proofs/commutative-distallation-proof.tex
\begin{proof}
 
We list here the equality cases omitted from the main proof in the paper.

\begin{itemize}
\item \case{$\mac{ \skeval }{ \skframe }{ \code\codetwo }{ \stack }{ \genv }
             \tomachasubone
             \mac{ \skeval }{ \skframe }{ \code }{ \codetwo\cons\stack }{ \genv }$}
     \[\begin{array}{lllllllll}
     &
     \decode{\mac{ \skeval }{ \skframe }{ \code\codetwo }{ \stack }{ \genv } }
     &=&
     \decodecompat{\skframe}{\genv}\ctxholep{\dstackp{\code\codetwo}}
     &=&
     \decodecompat{\skframe}{\genv}\ctxholep{\decodep{\codetwo\cons\stack}{\code}}
     &=&
     \decode{ \mac{ \skeval }{ \skframe }{ \code }{ \codetwo\cons\stack }{ \genv } }
 \end{array}\]

\item \case{$\mac{ \skeval }{ \skframe }{ \l\var.\code }{ \stempty }{ \genv }
             \tomachasubtwo
             \mac{ \skeval }{ \var\cons\skframe }{ \code }{ \stempty }{ \openscopem\var\cons\genv }$}
    \[\begin{array}{lllllllll}
    &
    \decode{\mac{ \skeval }{ \skframe }{ \l\var.\code }{ \stempty }{ \genv }}
    &=&
    \decodecompat{\skframe}{\genv}\ctxholep{\l\var.\code}
    \\
    &&=&
    \decodecompat{ (\var\cons\skframe) }{ (\openscopem\var\cons\genv) }\ctxholep{\code}
    &=&
    \decode{\mac{ \skeval }{ \var\cons\skframe }{ \code }{ \stempty }{ \openscopem\var\cons\genv }}
    \end{array}\]

\item \case{$\mac{ \skeval }{ \skframe }{ \var }{ \stack }{ \genv }
             \tomachasubthree
             \mac{ \skback }{ \skframe }{ \var }{ \stack }{ \genv }$}
    \[\begin{array}{lllllllll}
     &
     \decode{\mac{ \skeval }{ \skframe }{ \var }{ \stack }{ \genv }}
     &=&
     \decodecompat{\skframe}{\genv}\ctxholep{\dstackp{\var}}
     &=&
     \decode{\mac{ \skback }{ \skframe }{ \var }{ \stack }{ \genv }}
     \end{array}\]

\item \case{$\mac{ \skback }{ \skap{\code}{\stack}\cons\skframe }{ \codetwo }{ \stempty }{ \genv }
             \tomachasubfive
             \mac{ \skback }{ \skframe }{ \code\codetwo }{ \stack }{ \genv }$}
    $$
     \decode{\mac{ \skback }{ \skap{\code}{\stack}\cons\skframe }{ \codetwo }{ \stempty }{ \genv }}
     =
     \decodecompat{\skap{\code}{\stack}\cons\skframe}{\genv}\ctxholep{\codetwo}
     =
     \decodecompat{\skframe}{\genv}\ctxholep{\dstackp{\code\,\codetwo}}
     =
     \decode{\mac{ \skback }{ \skframe }{ \code\codetwo }{ \stack }{ \genv }}
    $$
\item \case{$\mac{ \skback }{ \skframe }{ \code }{ \codetwo\cons\stack }{ \genv }
             \tomachasubsix
             \mac{ \skeval }{ \skap{\code}{\stack}\cons\skframe }{ \codetwo }{ \stempty }{ \genv }$}
    \[\begin{array}{lllllllll}
     
     \decode{\mac{ \skback }{ \skframe }{ \code }{ \codetwo\cons\stack }{ \genv }}
     &=&
     \decodecompat{\skframe}{\genv}\ctxholep{\decode{\codetwo\cons\stack}\ctxholep{\code}}\\
     &=&
     \decodecompat{\skframe}{\genv}\ctxholep{\decode{\stack}\ctxholep{\code\,\codetwo}}
     \\
     &=&
     \decodecompat{ (\skap{\code}{\stack}\cons\skframe) }{\genv}\ctxholep{\codetwo}
     &=&
     \decode{\mac{ \skeval }{ \skap{\code}{\stack}\cons\skframe }{ \codetwo }{ \stempty }{ \genv }}
     \end{array}\]\qed
\end{itemize}
\end{proof}

%% file: proofs/subterm-invariant-proof.tex
The subterm invariant as formulated in the paper is a consequence of the last point of the following more general invariant, because $\tomache$ duplicates codes from the environment, here proved to be subterms of the initial term.

\begin{lemma}[Subterm Invariant]
\label{l:subterm-invariant-app} 
Let $\state = \skamstate{\skphase}{\skframe}{\codetwo}{\stack}{\genv}$ be a
state reachable from the initial code $\code$.
Then
  \begin{enumerate}
  \item \emph{Evaluating Code}: \label{p:subterm-invariant-one}
  if $\skphase = \skeval$, then $\codetwo$ is a subterm of $\code$;
  
  \item \emph{Stack}:  \label{p:subterm-invariant-two}
  any code  in the stack $\stack$ is a subterm of $\code$;
  
  \item \emph{Frame}: \label{p:subterm-invariant-three}
  if $\skframe = \skframetwo\cons\skap{\codethree}{\stacktwo}\cons\skframethree$, then any code in $\stacktwo$ is a subterm of $\code$;
  
  \item \emph{Global Environment}: \label{p:subterm-invariant-four}
  if $\genv = \genvtwo\cons\esub{\var}{\codethree}\cons\genvthree$, then $\codethree$ is a subterm of $\code$;
  \end{enumerate}
\end{lemma}

\begin{proof}
Let us use $\code_0$ for the initial term.
The invariant trivially holds for the initial state
  $\skamstate{\skeval}{\stempty}{\code_0}{\stempty}{\stempty}$.
  In the inductive case we look at the last transition:
  \begin{itemize}

  \item \case{$
          \mac{\skeval}{\skframe}{\code\codetwo}{\stack}{\genv}
          \tomachasubone
          \mac{\skeval}{\skframe}{\code}{\codetwo\cons\stack}{\genv}
        $}

      \begin{enumerate}
      \item \emph{Evaluating Code}: By \ih, $\code\codetwo$ is a subterm of $\code_0$, so $\code$ is also a subterm of $\code_0$.
      \item \emph{Stack}: by \ih, $\code\codetwo$ is a subterm of $\code_0$, so $\codetwo$ is also a subterm of $\code_0$.
            Moreover, any piece of code in $\stack$ is a subterm of $\code_0$ by \ih.
      \item \emph{Frame}: it follows from the \ih, since the frame $\skframe$ is unchanged.
      \item \emph{Environment}: it follows from the \ih, since the environment $\genv$ is unchanged.
      \end{enumerate}

  \item \case{$
          \mac{\skeval}{\skframe}{\l\var.\code}{\codetwo\cons\stack}{\genv}
          \tomachm
          \mac{\skeval}{\skframe}{\code}{\stack}{\esub{\var}{\codetwo}\cons\genv}
        $}
      \begin{enumerate}
      \item \emph{Evaluating Code}: note that $\code$ is a subterm of $\l\var.\code$.
      \item \emph{Stack}: note that any piece code in $\stack$ is also in $\codetwo\cons\stack$.
      \item \emph{Frame}: it follows from the \ih, since $\skframe$ is not modified.
      \item \emph{Environment}: the new environment is of the form $\esub{\var}{\codetwo}\cons\genv$.
            Pieces of code in $\genv$ are subterms of $\code_0$ by \ih.
            Moreover $\codetwo$ is the top of the stack $\codetwo\cons\stack$ so it is also
            a subterm of $\code_0$.
      \end{enumerate}

  \item \case{$
          \mac{\skeval}{\skframe}{\l\var.\code}{\stempty}{\genv}
          \tomachasubtwo
          \mac{\skeval}{\var\cons\skframe}{\code}{\stempty}{\openscopem\var\cons\genv}
        $}
    
      \begin{enumerate}
      \item \emph{Evaluating Code}: note that $\code$ is a subterm of $\l\var.\code$ which is in turn a subterm of $\code_0$ by \ih.
      \item \emph{Stack}: trivial since the stack $\stack$ is empty.
      \item \emph{Frame}: any pair of the form $\skap{\codetwo}{\stacktwo}$ in the frame $\var\cons\skframe$ is also
            already present in $\skframe$, so by \ih\ any piece of code in $\stacktwo$ is a subterm of $\code_0$.
      \item \emph{Environment}: it follows from the \ih, since the environment $\genv$ is unchanged.
      \end{enumerate}

  \item \case{$
          \mac{\skeval}{\skframe}{\var}{\stack}{\genv}
          \tomache
          \mac{\skeval}{\skframe}{\rename{\code}}{\stack}{\genv}
        $}
      \begin{enumerate}
      \item \emph{Evaluating Code}: note that $\code$ is bound by $\genv$. By \ih, it is a subterm of $\code_0$.
            So $\rename{\code}$ is also a subterm of $\code_0$.
      \item \emph{Stack}: it follows from the \ih, since the stack $\stack$ is unchanged.
      \item \emph{Frame}: it follows from the \ih, since the frame $\skframe$ is unchanged.
      \item \emph{Environment}: it follows from the \ih, since the environment $\genv$ is unchanged.
      \end{enumerate}

  \item \case{$
          \mac{\skeval}{\skframe}{\var}{\stack}{\genv}
          \tomachasubthree
          \mac{\skback}{\skframe}{\var}{\stack}{\genv}
        $}
      \begin{enumerate}
      \item \emph{Evaluating Code}: trivial since $\skphase \neq \skeval$.
      \item \emph{Stack}: it follows from the \ih, since the stack $\stack$ is unchanged.
      \item \emph{Frame}: it follows from the \ih, since the frame $\skframe$ is unchanged.
      \item \emph{Environment}: it follows from the \ih, since the environment $\genv$ is unchanged.
      \end{enumerate}

  \item \case{$
          \mac{\skback}{\var\cons\skframe}{\code}{\stempty}{\genv}
          \tomachasubfour
          \mac{\skback}{\skframe}{\l\var.\code}{\stempty}{\closescopem\var\cons\genv}
        $}

      \begin{enumerate}
      \item \emph{Evaluating Code}: trivial since $\skphase \neq \skeval$.
      \item \emph{Stack}: trivial since the stack is empty.
      \item \emph{Frame}: any pair of the form $\skap{\codetwo}{\stack}$ in the frame $\skframe$ is also in the frame $\var\cons\skframe$,
            so any piece of code in $\stack$ is a subterm of $\code_0$ by \ih.
      \item \emph{Environment}: any substitution of the form $\esub{\vartwo}{\codetwo}$ in the environment $\closescopem\var\cons\genv$
            is also in the environment $\genv$, so $\codetwo$ is a subterm of $\code_0$ by \ih.
      \end{enumerate}

  \item \case{$
          \mac{\skback}{\skap{\code}{\stack}\cons\skframe}{\codetwo}{\stempty}{\genv}
          \tomachasubfive
          \mac{\skback}{\skframe}{\code\codetwo}{\stack}{\genv}
        $}

      \begin{enumerate}
      \item \emph{Evaluating Code}: trivial since $\skphase \neq \skeval$.
      \item \emph{Stack}: the stack $\stack$ occurs at the left-hand side in the frame $\skap{\code}{\stack}\cons\skframe$,
            so by \ih\ we know that any piece of code in $\stack$ is a subterm of $\code_0$.
      \item \emph{Frame}: any pair $\skap{\codethree}{\stack}$ in the frame $\skframe$ is also
            in the frame $\skap{\code}{\stack}\cons\skframe$, so any piece of code
            in $\stack$ must be a subterm of $\code_0$.
      \item \emph{Environment}: it follows from the \ih, since the environment $\genv$ is unchanged.
      \end{enumerate}

  \item \case{$
          \mac{\skback}{\skframe}{\code}{\codetwo\cons\stack}{\genv}
          \tomachasubsix
          \mac{\skeval}{\skap{\code}{\stack}\cons\skframe}{\codetwo}{\stempty}{\genv}
        $}
      \begin{enumerate}
      \item \emph{Evaluating Code}: note that $\codetwo$ is an element of the stack at the left-hand side of the transition,
            so by \ih\ $\codetwo$ is a subterm of $\code_0$.
      \item \emph{Stack}: trivial since the stack is empty.
      \item \emph{Frame}: any pair in the frame $\skap{\code}{\stack}\cons\skframe$ is also in the frame $\skframe$
            except for $\skap{\code}{\stack}$.
            Consider a piece of code $\codefour$ in the stack $\stack$. It is trivially also a piece of
            code in the stack $\codetwo\cons\stack$, so by \ih\ we have that $\codefour$ is
            a subterm of $\code_0$.
      \item \emph{Environment}: it follows from the \ih, since the environment $\genv$ is unchanged.\qed
      \end{enumerate} 
  \end{itemize}
\end{proof}

%% file: bisimulation_proof.tex
\withproofs{

We first need an auxiliary lemma (\reflemma{eqstruct_commute_ilo}), which uses an alternative, inductive definition of \lo\ contexts:
\begin{definition}[\ilo\ Contexts]
  A context $\ctx$ is \emph{inductively \lo} (or \ilo) if a judgment about it can be derived using the following inductive rules:    
  \begin{center}
$\begin{array}{cccccc}
	\AxiomC{}
	\RightLabel{(ax-\ilo)}
	\UnaryInfC{$\ctxhole$ is \ilo}
	\DisplayProof
	&
	\AxiomC{$\ctx$ is \ilo}
	\AxiomC{$\ctx\neq\sctxp{\la\var\ctxtwo}$}
	\RightLabel{(@l-\ilo)}
	\BinaryInfC{$\ctx \tm$ is \ilo}
	\DisplayProof \\\\
	
		\AxiomC{$\ctx$ is \ilo}
	\RightLabel{($\l$-\ilo)}
	\UnaryInfC{$\la\var\ctx$ is \ilo}
	\DisplayProof 
	&
		\AxiomC{$\tm$ is \quiet}
		\AxiomC{$\ctx$ is \ilo}
		\RightLabel{(@r-\ilo)}
		\BinaryInfC{$\tm \ctx$ is \ilo}
		\DisplayProof 
		 \\\\
	\multicolumn{2}{c}{	  
	\AxiomC{$\ctx$ is \ilo}
	\AxiomC{$\var\notin\lfv\ctx$}
	\RightLabel{(ES-\ilo)}
	\BinaryInfC{$\ctx \esub\var\tm$ is \ilo}
	\DisplayProof
	}
\end{array}$
\end{center}
\end{definition}

\begin{lemma}
\label{l:LO-characts-ind}
A context $\ctx$ is \ilo\ iff it is \lo.
\end{lemma}
\begin{proof}
	An immediate induction on $\ctx$.\qed
\end{proof}

\begin{lemma}
\label{l:eqstruct_commute_ilo}
If $\ctx$ is a \lo\ context and $\ctx$ does not bind any of the variables in $\fv{\tmtwo}$,
then $\ctxp{\tm\esub{\var}{\tmtwo}} \eqstruct \ctxp{\tm}\esub{\var}{\tmtwo}$.
\end{lemma}
\begin{proof}
A context is \lo\ iff it is \ilo\ (\reflemma{LO-characts-ind}).
The property is then proved by induction on the derivation that $\ctx$ is an \ilo\ context.\qed
\end{proof}

\begin{proof}[Structural Equivalence $\eqstruct$ is a Strong Bisimulation, \refprop{bisimulation}]

Let $\tostructsym$ be the symmetric and contextual closure of the axioms by which
$\eqstruct$ is defined, \ie 
$$
\begin{array}{rll@{\hspace{1em}}l}
    \tm\esub{\var}{\tmtwo}                         & \tostructgc  & \tm                                            & \text{if $\var \not\in \fv{\tm}$} \\
    \tm\esub{\var}{\tmtwo}\esub{\vartwo}{\tmthree} & \tostructcom & \tm\esub{\vartwo}{\tmthree}\esub{\var}{\tmtwo} & \text{if $\vartwo \not\in \fv{\tmtwo}$ and $\var \not\in \fv{\tmthree}$} \\
    \tm\esub{\var}{\tmtwo}\esub{\vartwo}{\tmthree} & \tostructes  & \tm\esub{\var}{\tmtwo\esub{\vartwo}{\tmthree}} & \text{if $\vartwo \not\in \fv{\tm}$} \\
    \tm\esub{\var}{\tmtwo}                         & \tostructdup & \varsplit{\tm}{\var}{\vartwo}\esub{\var}{\tmtwo}\esub{\vartwo}{\tmtwo}   \\
    (\l\var.\tm)\esub{\vartwo}{\tmtwo}             & \tostructlam & \l\var.\tm\esub{\vartwo}{\tmtwo}               & \text{if $\var \not\in \fv{\tmtwo}$} \\
    (\tm\,\tmtwo)\esub{\var}{\tmthree}             & \tostructapl & \tm\esub{\var}{\tmthree}\,\tmtwo               & \text{if $\var \not\in \fv{\tmtwo}$} \\
    (\tm\,\tmtwo)\esub{\var}{\tmthree}             & \tostructapr & \tm\,\tmtwo\esub{\var}{\tmthree}               & \text{if $\var \not\in \fv{\tm}$} \\
\end{array}
$$
Note that $\eqstruct$ is the reflexive-transitive closure
of $\tostructsym$. It suffices to show that $\tostructsym\tolo\ \subseteq\ \tolo\eqstruct$, preserving the kind of step (multiplicative/exponential).
The fact that $\tostructsym^*$ is a bisimulation then follows by induction on the number
of $\tostructsym$ steps.

Let $\tmthree \tostructsym \tm \tolo \tmtwo$. The proof of $\tmthree \tolo \eqstruct \tmtwo$
goes by induction on the context under which the step $\tm \tolo \tmtwo$ takes place. In the following proof note that:
\begin{enumerate}
\item $\tom$ steps are sent to $\tom$ steps,  
\item $\toe$ steps are sent to $\toe$ steps, and 
\item no step is ever duplicated.
\end{enumerate}

Cases:
\begin{enumerate}

\item \casealt{Base case 1: multiplicative root step, $\tm = \sctxp{\l\var.\tmp}\tmtwop \rtodb \sctxp{\tmp\esub{\var}{\tmtwop}}$}
  If the $\tostructsym$ step is internal to $\tmp$, internal to $\tmtwop$, or internal to the argument of one of the substitutions in $\sctx$,
  then the pattern of the $\tostructsym$ redex does not overlap with the $\rtodb$ step, and the proof
  is immediate, as the two steps commute.
  Otherwise, we consider every possible case of $\tostructsym$:
  \begin{enumerate}
  \item \caselight{Garbage collection, $\tostructgc$.}
    The garbage collected substitution must be one of the substitutions in
    $\sctx$, \ie\ $\sctx$ must be of the form $\sctxtwop{\sctxthree\esub{\vartwo}{\tmthreep}}$.
    Then:
    $$
    \commutesdbEEABCD{\tostructgc}{\tostructgc}{
       \sctxtwop{\sctxthreep{ \l\var.\tmp }\esub{\vartwo}{\varthree}} \tmtwop
    }{
       \sctxtwop{\sctxthreep{ \tmp\esub{\var}{ \tmtwop } }\esub{\vartwo}{\varthree}}
    }{
       \sctxtwop{\sctxthreep{ \l\var.\tmp }} \tmtwop
    }{
       \sctxtwop{\sctxthreep{ \tmp\esub{\var}{ \tmtwop } }}
    }
    $$
  \item \caselight{Commutation of independent substitutions, $\tostructcom$.}
    The substitutions that are commuted must be both in $\sctx$,
    \ie\ $\sctx$ must be of the form $\sctxtwop{\sctxthree\esub{\vartwo}{\tmthreep}\esub{\varthree}{\tmfourp}}$.
    Then:
    $$
    \commutesdbEEABCD{\tostructcom}{\tostructcom}{
       \sctxtwop{\sctxthreep{ \l\var.\tmp }\esub{\vartwo}{\tmthreep}\esub{\varthree}{\tmfourp}} \tmtwop
    }{
       \sctxtwop{\sctxthreep{ \tmp\esub{\var}{ \tmtwop } }\esub{\vartwo}{\tmthreep}\esub{\varthree}{\tmfourp}}
    }{
       \sctxtwop{\sctxthreep{ \l\var.\tmp }\esub{\varthree}{\tmfourp}\esub{\vartwo}{\tmthreep}} \tmtwop
    }{
       \sctxtwop{\sctxthreep{ \tmp\esub{\var}{ \tmtwop } }\esub{\varthree}{\tmfourp}\esub{\vartwo}{\tmthreep}}
    }
    $$
  \item \caselight{Composition of substitutions, $\tostructes$.}
    The substitutions that are composed must be both in $\sctx$,
    \ie\ $\sctx$ must be of the form $\sctxtwop{\sctxthree\esub{\vartwo}{\tmthreep}\esub{\varthree}{\tmfourp}}$.
    Then:
    $$
    \commutesdbEEABCD{\tostructes}{\tostructes}{
       \sctxtwop{\sctxthreep{ \l\var.\tmp }\esub{\vartwo}{\tmthreep}\esub{\varthree}{\tmfourp}} \tmtwop
    }{
       \sctxtwop{\sctxthreep{ \tmp\esub{\var}{ \tmtwop } }\esub{\vartwo}{\tmthreep}\esub{\varthree}{\tmfourp}}
    }{
       \sctxtwop{\sctxthreep{ \l\var.\tmp }\esub{\vartwo}{\tmthreep\esub{\varthree}{\tmfourp}}} \tmtwop
    }{
       \sctxtwop{\sctxthreep{ \tmp\esub{\var}{ \tmtwop } }\esub{\vartwo}{\tmthreep\esub{\varthree}{\tmfourp}}}
    }
    $$
  \item \caselight{Duplication, $\tostructdup$.}
    The duplicated substitution must be one of the substitutions in $\sctx$,
    \ie\ $\sctx$ must be of the form $\sctxtwop{\sctxthree\esub{\vartwo}{\tmthreep}}$.
    Then:
    $$
    \commutesdbEEABCD{\tostructdup}{\tostructdup}{
       \sctxtwop{\sctxthreep{ \l\var.\tmp }\esub{\vartwo}{\tmthreep}} \tmtwop
    }{
       \sctxtwop{\sctxthreep{ \tmp\esub{\var}{ \tmtwop } }\esub{\vartwo}{\tmthreep}}
    }{
       \sctxtwop{\varsplit{(\sctxthreep{ \l\var.\tmp })}{\vartwo}{\varthree}\esub{\vartwo}{\tmthreep}\esub{\varthree}{\tmthreep}} \tmtwop
    }{
       \sctxtwop{\varsplit{(\sctxthreep{ \tmp\esub{\var}{ \tmtwop } })}{\vartwo}{\varthree}\esub{\vartwo}{\tmthreep}\esub{\varthree}{\tmthreep}}
    }
    $$
  \item \caselight{Commutation with abstraction, $\tostructlam$.}
    The commuted substitution must be the innermost substitution in $\sctx$,
    \ie\ $\sctx$ must be of the form $\sctxtwop{\esub{\vartwo}{\tmthreep}}$, and:
      $$
      \commuteslsEEABCD{\tostructlam}{\tostructcom}{
         \sctxtwop{ (\l\var.\tmp)\esub{\vartwo}{\tmthreep} } \tmtwop
      }{
         \sctxtwop{ \tmp\esub{\var}{\tmtwop}\esub{\vartwo}{\tmthreep} }
      }{
         \sctxtwop{ \l\var.\tmp\esub{\vartwo}{\tmthreep} } \tmtwop
      }{
         \sctxtwop{ \tmp\esub{\vartwo}{\tmthreep}\esub{\var}{\tmtwop} }
      }
      $$
      Note that the diagram can be also read from the bottom-up for a
      reverse application of the $\tostructlam$ rule.
      In order to be able to apply $\tostructcom$,
      note that $\var \not\in \fv{\tmthreep}$ by application of the $\tostructlam$ rule,
      and that $\vartwo \not\in \fv{\tmtwop}$ by the bound variable convention.
  \item \caselight{Left commutation with application, $\tostructapl$.}
    The only possibility is that the outermost substitution of $\sctx$ commutes with
    the application taking part in the $\tom$ step.
    That is, $\sctx$ must be of the form $\sctxtwo\esub{\vartwo}{\tmthreep}$ and:
    $$
    \commutesdbEEABCD{\tostructapl}{=}{
       \sctxtwop{ \l\var.\tmp }\esub{\vartwo}{\tmthreep} \tmtwop
    }{
       \sctxtwop{ \tmp\esub{\var}{\tmtwop} }\esub{\vartwo}{\tmthreep} \tmtwop
    }{
       (\sctxtwop{ \l\var.\tmp }\,\tmtwop)\esub{\vartwo}{\tmthreep}
    }{
       \sctxtwop{ \tmp\esub{\var}{\tmtwop} }\esub{\vartwo}{\tmthreep}
    }
    $$
  \item \caselight{Right commutation with application, $\tostructapr$.}
    Note that every $\tostructapr$ (and ${\tostructapr}^{-1}$) redex in $(\l\var.\tmp)\sctx\,\tmtwop$
    must be internal to either $\tmp$, $\tmtwop$, or the argument of one of the substitutions
    in $\sctx$. We have already argued that in these cases the steps commute.
  \end{enumerate}

\item \casealt{Base case 2: exponential root step, $\tm = \ctxp{\var}\esub{\var}{\tmp} \rtols \ctxp{\tmp}\esub{\var}{\tmp}$}

  If the substitution that is contracted by the exponential step does not
  take part in the pattern of the $\tostructsym$ step, it is immediate to check
  that the property holds. More precisely, suppose that $\ctxp{\var}\esub{\var}{\tmp} \tostructsym \ctxtwop{\var}\esub{\var}{\tmpp}$,
  where $\ctxtwo$ and $\tmpp$ result respectively from $\ctx$ and $\tm$ by
  a single step of $\tostructsym$.
  Note that we have that either $\ctx \tostructsym \ctxtwo$ and $\tmp = \tmpp$
  or vice-versa. Then:
  $$
  \commuteslsEEABCD{\tostructsym}{\tostructsym^*}{
     \ctxp{\var}\esub{\var}{\tmp}
  }{
     \ctxp{\tmp}\esub{\var}{\tmp}
  }{
     \ctxtwop{\var}\esub{\var}{\tmpp}
  }{
     \ctxtwop{\tmpp}\esub{\var}{\tmpp}
  }
  $$
  Note that when commutation affects $\tmp$ (\ie\ if we are in the case in which
  $\ctx = \ctxtwo$ and $\tmp \tostructsym \tmpp$), then the right-hand side of the
  diagram must be closed by two $\tostructsym$ steps: one for each copy of $\tmp$.

  So we may assume that the substitution that is contracted by the exponential step does
  take part in the pattern of the $\tostructsym$ step. We consider every possible case of
  $\tostructsym$.

  \begin{enumerate}
  \item \caselight{Garbage collection, $\tostructgc$.}
    The garbage collected substitution cannot erase the contracted occurrence of
    $\var$, since $\ctx$ is a \lo\ context, and it cannot go inside substitutions.
    Two subcases, depending on the position of the hole of $\ctx$ with respect to the
    node of the garbage collected substitution:
    \begin{enumerate}
    \item If the hole of $\ctx$ lies inside the body of the garbage collected substitution,
          \ie\ $\ctx = \ctxtwop{\ctxthree\esub{\vartwo}{\tmtwop}}$ with $\vartwo \notin \fv{\ctxthreep{\var}}$,
          then:
      $$
      \commuteslsEEABCD{\tostructgc}{\tostructgc}{
         \ctxtwop{\ctxthreep{\var}\esub{\vartwo}{\tmtwop}}\esub{\var}{\tmp}
      }{
         \ctxtwop{\ctxthreep{\tmp}\esub{\vartwo}{\tmtwop}}\esub{\var}{\tmp}
      }{
         \ctxtwop{\ctxthreep{\var}}\esub{\var}{\tmp}
      }{
         \ctxtwop{\ctxthreep{\tmp}}\esub{\var}{\tmp}
      }
      $$
      Note that $\vartwo \notin \fv{\ctxthreep{\tmp}}$ since we may assume that
      $\vartwo \notin \fv{\tmp}$ by the bound variable convention.
    \item
      Otherwise, the hole of $\ctx$ must be disjoint from the node of the garbage collected
      substitution, \ie\ there must be a two-hole context $\ctxtwo$ such that:
      $$
        \ctx = \ctxtwop{\ctxhole,\tmtwop\esub{\vartwo}{\tmthreep}}
      $$
      where $\vartwo \not\in \fv{\tmtwop}$. Then:
      $$
      \commuteslsEEABCD{\tostructgc}{\tostructgc}{
         \ctxtwop{\var,\tmtwop\esub{\vartwo}{\tmthreep}}\esub{\var}{\tmp}
      }{
         \ctxtwop{\tmp,\tmtwop\esub{\vartwo}{\tmthreep}}\esub{\var}{\tmp}
      }{
         \ctxtwop{\var,\tmtwop}\esub{\var}{\tmp}
      }{
         \ctxtwop{\tmp,\tmtwop}\esub{\var}{\tmp}
      }
      $$
    \end{enumerate}
  \item \caselight{Commutation of independent substitutions, $\tostructcom$.}
    \label{l:smam_bisimulation_exponential_tostructcom}
    Note that the contracted occurrence of $\var$ cannot be inside the argument
    of any of the commuted substitutions, since $\ctx$ is a \lo\ context
    and it cannot go inside substitutions.
    Since the contracted substitution is commuted, we have that $\ctx$ must be of the form
    $\ctxtwo\esub{\vartwo}{\tmtwop}$ and the situation is:
      $$
      \commuteslsEEABCD{\tostructcom}{\tostructcom}{
         \ctxtwop{\var}\esub{\vartwo}{\tmtwop}\esub{\var}{\tmp}
      }{
         \ctxtwop{\tmp}\esub{\vartwo}{\tmtwop}\esub{\var}{\tmp}
      }{
         \ctxtwop{\var}\esub{\var}{\tmp}\esub{\vartwo}{\tmtwop}
      }{
         \ctxtwop{\tmp}\esub{\var}{\tmp}\esub{\vartwo}{\tmtwop}
      }
      $$
  \item \caselight{Composition of substitutions, $\tostructes$.}
    Note that the contracted occurrence of $\var$ cannot be inside the argument
    of any of the two substitutions that take part in the $\tostructes$ step,
    since $\ctx$ is a \lo\ context and it cannot go inside substitutions.
    We know that the contracted substitution takes part in the $\tostructes$ step.
    We consider two subcases, depending on whether the $\tostructes$ rule is applied
    from left to right or from right to left, since the situation is not symmetrical.
    \begin{enumerate}
    \item
      If the $\tostructes$ step is applied from left to right, then $\ctx$ must be of the
      form $\ctxtwo\esub{\vartwo}{\tmtwop}$ with $\var \not\in \fv{\ctxtwop{\var}}$.
      This is a contradiction, so this case is not actually possible.
    \item
      \label{l:smam_bisimulation_exponential_tostructes}
      If the $\tostructes$ step is applied from right to left,
      then $\tmp$ must be of the form $\tmpp\esub{\vartwo}{\tmtwop}$ and:
        $$
        \commuteslsEEABCD{\tostructes}{\eqstruct}{
           \ctxp{\var}\esub{\var}{\tmpp\esub{\vartwo}{\tmtwop}}
        }{
           \ctxp{\tmpp\esub{\vartwo}{\tmtwop}}\esub{\var}{\tmpp\esub{\vartwo}{\tmtwop}}
        }{
           \ctxp{\var}\esub{\var}{\tmpp}\esub{\vartwo}{\tmtwop}
        }{
           \ctxp{\tmpp}\esub{\var}{\tmpp}\esub{\vartwo}{\tmtwop}
        }
        $$
        To close the right-hand side of the diagram, we are left to show that:
        $$
           \ctxp{\tmpp\esub{\vartwo}{\tmtwop}}\esub{\var}{\tmpp\esub{\vartwo}{\tmtwop}}
           \eqstruct
           \ctxp{\tmpp}\esub{\var}{\tmpp}\esub{\vartwo}{\tmtwop}
        $$
        First note that $\ctx$ is a \lo\ context, and that, by the bound variable convention, $\ctx$ does
        not bind any of the variables in $\fv{\tmtwop}$. By resorting to Lemma~\ref{l:eqstruct_commute_ilo},
        this allows us to commute the substitution that:
        $$
           \begin{array}{lll}
               & \ctxp{\tmpp\esub{\vartwo}{\tmtwop}}\esub{\var}{\tmpp\esub{\vartwo}{\tmtwop}} \\
           \eqstruct
               & \ctxp{\tmpp}\esub{\vartwo}{\tmtwop}\esub{\var}{\tmpp\esub{\vartwo}{\tmtwop}}
               & \text{ by Lemma~\ref{l:eqstruct_commute_ilo}} \\
           \tostructes
               & \ctxp{\tmpp}\esub{\vartwo}{\tmtwop}\esub{\var}{\tmpp}\esub{\vartwo}{\tmtwop} \\
           =
               & \ctxp{\tmpp}\esub{\vartwo}{\tmtwop}\esub{\var}{\tmpp\isub{\vartwo}{\varthree}}\esub{\varthree}{\tmtwop}
               & \text{ renaming $\vartwo$ to $\varthree$} \\
           \tostructcom
               & \ctxp{\tmpp}\esub{\var}{\tmpp\isub{\vartwo}{\varthree}}\esub{\vartwo}{\tmtwop}\esub{\varthree}{\tmtwop} \\
           \tostructdup
               & \ctxp{\tmpp}\esub{\var}{\tmpp}\esub{\vartwo}{\tmtwop}
           \end{array}
        $$
    \end{enumerate}
  \item \caselight{Duplication, $\tostructdup$.}
    Note that the contracted occurrence of $\var$ cannot be inside the argument
    of any of the two substitutions that take part in the $\tostructdup$ step,
    since $\ctx$ is a \lo\ context and it cannot go inside substitutions.
    We consider two cases, depending on whether $\tostructdup$ is applied from left to right
    or from right to left:
    \begin{enumerate}
    \item
          From left to right: the contracted occurrence of $\var$ is either renamed to $\vartwo$
          or left untouched as $\var$. Let $\varthree$ denote $\var$ or $\vartwo$, correspondingly.
          In both cases we have:
          $$
          \commuteslsEEABCD{\tostructdup}{\tostructdup}{
             \ctxp{\var}\esub{\var}{\tmp}
          }{
             \ctxp{\tmp}\esub{\var}{\tmp}
          }{
             \varsplit{\ctx}{\var}{\vartwo}\ctxholep{\varthree}\esub{\var}{\tmp}\esub{\vartwo}{\tmp}
          }{
             \varsplit{\ctx}{\var}{\vartwo}\ctxholep{\tmp}\esub{\var}{\tmp}\esub{\vartwo}{\tmp}
          }
          $$
    \item
          From right to left:
          then $\ctx$ is of the form $\varsplit{\ctxtwo}{\vartwo}{\var}\esub{\vartwo}{\tmp}$,
          where $\ctxtwo$ has no occurrences of $\var$, and:
          $$
          \commuteslsEEABCD{\tostructdup}{\tostructdup}{
             \varsplit{\ctxtwo}{\vartwo}{\var}\ctxholep{\var}\esub{\vartwo}{\tmp}\esub{\var}{\tmp}
          }{
             \varsplit{\ctxtwo}{\vartwo}{\var}\ctxholep{\tmp}\esub{\vartwo}{\tmp}\esub{\var}{\tmp}
          }{
             \ctxtwo\ctxholep{\vartwo}\esub{\vartwo}{\tmp}
          }{
             \ctxtwo\ctxholep{\tmp}\esub{\vartwo}{\tmp}
          }
          $$
    \end{enumerate}
  \item \caselight{Commutation with abstraction, $\tostructlam$.}
    \label{l:smam_bisimulation_exponential_tostructlam}
    Then $\ctx$ is of the form $\l\vartwo.\ctxtwo$ and:
    $$
    \commuteslsEEABCD{\tostructlam}{\tostructlam}{
       (\l\vartwo.\ctxtwop{\var})\esub{\var}{\tmp}
    }{
       (\l\vartwo.\ctxtwop{\tmp})\esub{\var}{\tmp}
    }{
       \l\vartwo.\ctxtwop{\var}\esub{\var}{\tmp}
    }{
       \l\vartwo.\ctxtwop{\tmp}\esub{\var}{\tmp}
    }
    $$
  \item \caselight{Left commutation with application, $\tostructapl$.}
    \label{l:smam_bisimulation_exponential_tostructapl}
    Then $\ctx$ is of the form $\ctx\,\tmtwop$ and:
    $$
    \commuteslsEEABCD{\tostructapl}{\tostructapl}{
       (\ctxp{\var}\,\tmtwop)\esub{\var}{\tmp}
    }{
       (\ctxp{\tmp}\,\tmtwop)\esub{\var}{\tmp}
    }{
       \ctxp{\var}\esub{\var}{\tmp}\,\tmtwop
    }{
       \ctxp{\tmp}\esub{\var}{\tmp}\,\tmtwop
    }
    $$
  \item \caselight{Right commutation with application, $\tostructapr$.}
    \label{l:smam_bisimulation_exponential_tostructapr}
    Then $\ctx$ is of the form $\tmtwop\,\ctx$ and:
    $$
    \commuteslsEEABCD{\tostructapr}{\tostructapr}{
       (\tmtwop\,\ctxp{\var})\esub{\var}{\tmp}
    }{
       (\tmtwop\,\ctxp{\tmp})\esub{\var}{\tmp}
    }{
       \tmtwop\,\ctxp{\var}\esub{\var}{\tmp}
    }{
       \tmtwop\,\ctxp{\tmp}\esub{\var}{\tmp}
    }
    $$
  \end{enumerate}

\item \casealt{Inductive case 1: inside an abstraction}
  Suppose that $\tm = \l\var.\tmp \to \l\var.\tmtwop = \tmtwo$.
  We consider two subcases, depending on whether the $\tostructsym$ step
  is internal to the body of the abstraction, or involves the outermost
  abstraction:
  \begin{enumerate}
  \item
    \label{l:smam_bisimulation_inside_abstraction_internal}
    If the application of the $\tostructsym$ step is internal to $\tmp$, we
    have by \ih:
      $$
      \commutesredEEABCD{\eqstruct}{\eqstruct}{
         \tmp
      }{
         \tmtwop
      }{
         \tmthreep
      }{
         \tmfourp
      }
      $$
    so is immediate to conclude that:
      $$
      \commutesredEEABCD{\eqstruct}{\eqstruct}{
         \l\var.\tmp
      }{
         \l\var.\tmtwop
      }{
         \l\var.\tmthreep
      }{
         \l\var.\tmfourp
      }
      $$
  \item
    If the outermost abstraction takes part in the $\tostructsym$ step,
    then a $\tostructlam$ step must have been applied,
    so $\tmp$ must be of the form $\tmpp\esub{\vartwo}{\tmtwop}$.
    We consider two further subcases, depending on whether the commuted substitution
    is involved in the reduction step:
    \begin{enumerate}
    \item 
      If the reduction step $\tmpp\esub{\vartwo}{\tmtwop} \to \tmthreep$ is an exponential,
      and the commuted substitution $\esub{\vartwo}{\tmtwop}$
      is the one contracted by the exponential step, then the situation
      is exactly like in case \ref{l:smam_bisimulation_exponential_tostructlam}
      ({\em Commutation with abstraction} for exponential steps), by reading
      the diagram from the bottom up.
    \item
      Otherwise, note that there cannot be a multiplicative step at the root,
      and that the step cannot be internal to $\tmtwop$, as \lo\ contexts do
      not go inside substitutions. Therefore the reduction step must be internal
      to $\tmpp$ and the situation is:
      $$
      \commutesredEEABCD{\tostructlam}{\tostructlam}{
         \l\var.\tmpp\esub{\vartwo}{\tmtwop}
      }{
         \l\var.\tmtwopp\esub{\vartwo}{\tmtwop}
      }{
         (\l\var.\tmpp)\esub{\vartwo}{\tmtwop}
      }{
         (\l\var.\tmtwopp)\esub{\vartwo}{\tmtwop}
      }
      $$
    \end{enumerate}
  \end{enumerate}
\item \casealt{Inductive case 2: left of an application} 
  Suppose that $\tm = \tmp\,\tmfive \to \tmtwop\,\tmfive = \tmtwo$.
  If the application of the $\tostructsym$ step is internal to $\tmp$,
  we may immediately conclude by \ih (analogous to case \ref{l:smam_bisimulation_inside_abstraction_internal}).
  The interesting case is when the outermost application takes part in the $\tostructsym$ step.
  There are two possibilities, depending on whether a $\tostructapl$ step or a $\tostructapr$ step
  is applied:
  \begin{enumerate}
  \item \casealt{$\tostructapl$ step}
    Then $\tmp$ must be of the form $\tmpp\esub{\var}{\tmthreep}$.
    We consider two further subcases, depending on whether the commuted substitution
    is involved in the reduction step:
    \begin{enumerate}
    \item
      If the reduction step $\tmpp\esub{\var}{\tmthreep} \to \tmfourp$ is an exponential step
      and the commuted substitution $\esub{\var}{\tmthreep}$ is also the one contracted by the
      exponential step, then the situation is exactly like in case
      \ref{l:smam_bisimulation_exponential_tostructapl} ({\em Left commutation with application}
      for exponential steps), by reading the diagram from the bottom up.
    \item
      Otherwise, note that the reduction step cannot be internal to $\tmthreep$, since \lo\ contexts
      do not go inside substitutions, so it must be internal to $\tmpp$ and the
      situation is:
        $$
        \commutesredEEABCD{\tostructapl}{\tostructapl}{
           \tmpp\esub{\var}{\tmthreep}\,\tmfive
        }{
           \tmtwopp\esub{\var}{\tmthreep}\,\tmfive
        }{
           (\tmpp\,\tmfive)\esub{\var}{\tmthreep}
        }{
           (\tmtwopp\,\tmfive)\esub{\var}{\tmthreep}
        }
        $$
    \end{enumerate}
  \item \casealt{$\tostructapr$ step}
    Then $\tmfive$ must be of the form $\tmfivep\esub{\var}{\tmthreep}$ and the situation is:
        $$
        \commutesredEEABCD{\tostructapr}{\tostructapr}{
           \tmp\,\tmfivep\esub{\var}{\tmthreep}
        }{
           \tmtwop\,\tmfivep\esub{\var}{\tmthreep}
        }{
           (\tmp\,\tmfivep)\esub{\var}{\tmthreep}
        }{
           (\tmtwop\,\tmfivep)\esub{\var}{\tmthreep}
        }
        $$
  \end{enumerate}

\item \casealt{Inductive case 3: right of an application} 
  Suppose that $\tm = \tmfive\,\tmp \to \tmfive\,\tmtwop = \tmtwo$.
  If the application of the $\tostructsym$ step is internal to $\tmp$,
  we may immediately conclude by \ih (analogous to case \ref{l:smam_bisimulation_inside_abstraction_internal}).
  The interesting case is when the outermost application takes part in the $\tostructsym$ step.
  There are two possibilities, depending on whether a $\tostructapl$ step or a $\tostructapr$ step
  is applied:
  \begin{enumerate}
  \item \casealt{$\tostructapl$ step}
    Then $\tmfive$ must be of the form $\tmfivep\esub{\var}{\tmthreep}$ and the situation is:
        $$
        \commutesredEEABCD{\tostructapl}{\tostructapl}{
           \tmfivep\esub{\var}{\tmthreep}\,\tmp
        }{
           \tmfivep\esub{\var}{\tmthreep}\,\tmtwop
        }{
           (\tmfivep\,\tmp)\esub{\var}{\tmthreep}
        }{
           (\tmfivep\,\tmtwop)\esub{\var}{\tmthreep}
        }
        $$
  \item \casealt{$\tostructapr$ step}
    Then $\tmp$ must be of the form $\tmpp\esub{\var}{\tmthreep}$.
    We consider two further subcases, depending on whether the commuted substitution
    is involved in the reduction step:
    \begin{enumerate}
    \item
      If the reduction step $\tmpp\esub{\var}{\tmthreep} \to \tmfourp$ is an exponential step
      and the commuted substitution $\esub{\var}{\tmthreep}$ is also the one contracted by the
      exponential step, then the situation is exactly like in case
      \ref{l:smam_bisimulation_exponential_tostructapr} ({\em Right commutation with application}
      for exponential steps), by reading the diagram from the bottom up.
    \item
      Otherwise, note that the reduction step cannot be internal to $\tmthreep$, since \lo\ contexts
      do not go inside substitutions, so it must be internal to $\tmpp$ and the
      situation is:
        $$
        \commutesredEEABCD{\tostructapr}{\tostructapr}{
           \tmfive\,\tmpp\esub{\var}{\tmthreep}
        }{
           \tmfive\,\tmtwopp\esub{\var}{\tmthreep}
        }{
           (\tmfive\,\tmpp)\esub{\var}{\tmthreep}
        }{
           (\tmfive\,\tmtwopp)\esub{\var}{\tmthreep}
        }
        $$
    \end{enumerate}
  \end{enumerate}

\item \casealt{Inductive case 4: left of a substitution} 
  Suppose that $\tm = \tmp\esub{\var}{\tmfive} \to \tmtwop\esub{\var}{\tmfive} = \tmtwo$.
  If the application of the $\tostructsym$ step is internal to $\tmp$,
  we may immediately conclude by \ih (analogous to case \ref{l:smam_bisimulation_inside_abstraction_internal}).
  The interesting case is when the outermost substitution node takes part in the $\tostructsym$ step.
  There are four possibilities,
  depending on whether a $\tostructgc$ step, a $\tostructcom$ step, a $\tostructes$ step, or a $\tostructdup$ step
  is applied:
  \begin{enumerate}
  \item \casealt{$\tostructgc$ step}
    The reduction step cannot be internal to $\tmfive$, since \lo\ contexts may not go
    inside substitutions, so the step must be internal to $\tmp$,
    and closing the diagram is trivial:
      $$
      \commutesredEEABCD{\tostructgc}{\tostructgc}{
         \tmp\esub{\var}{\tmfive}
      }{
         \tmtwop\esub{\var}{\tmfive}
      }{
         \tmp
      }{
         \tmtwop
      }
      $$
      Note that if $\var \not\in \fv{\tmp}$ then $\var \not\in \fv{\tmtwop}$ by the usual
      property that reduction does not create free variables.
  \item \casealt{$\tostructcom$ step}
    Then $\tmp$ must be of the form $\tmpp\esub{\vartwo}{\tmthreep}$ with $\var \not\in \fv{\tmthreep}$.
    We consider two further subcases, depending on whether the commuted substitution
    is involved in the reduction step:
    \begin{enumerate}
    \item
      If the reduction step $\tmpp\esub{\vartwo}{\tmthreep} \to \tmfourp$ is an exponential step
      and the commuted substitution $\esub{\vartwo}{\tmthreep}$ is also the one contracted by the
      exponential step, then the situation is exactly like in case
      \ref{l:smam_bisimulation_exponential_tostructcom} ({\em Commutation of independent substitutions}
      for exponential steps), by reading the diagram from the bottom up.
    \item
      Otherwise, note that the reduction step cannot be internal to $\tmthreep$, since \lo\ contexts
      may not go inside substitutions, so it must be internal to $\tmpp$, and the situation is:
        $$
        \commutesredEEABCD{\tostructcom}{\tostructcom}{
           \tmpp\esub{\vartwo}{\tmthreep}\esub{\var}{\tmfive}
        }{
           \tmtwopp\esub{\vartwo}{\tmthreep}\esub{\var}{\tmfive}
        }{
           \tmpp\esub{\var}{\tmfive}\esub{\vartwo}{\tmthreep}
        }{
           \tmtwopp\esub{\var}{\tmfive}\esub{\vartwo}{\tmthreep}
        }
        $$
    \end{enumerate}
  \item \casealt{$\tostructes$ step}
    Two cases, depending on whether the $\tostructes$ step is applied from left to right
    or from right to left:
    \begin{enumerate}
    \item \caselight{$\tostructes$ is applied from left to right.}
      Then $\tmp$ must be of the form $\tmpp\esub{\vartwo}{\tmthreep}$ with $\var \not\in \fv{\tmpp}$.
      We consider two further subcases, depending on whether the commuted substitution
      is involved in the reduction step:
      \begin{enumerate}
      \item
        If the reduction step $\tmpp\esub{\vartwo}{\tmthreep} \to \tmfourp$ is an exponential step
        and the commuted substitution $\esub{\vartwo}{\tmthreep}$ is also the one contracted by the
        exponential step, then the situation is exactly like in case
        \ref{l:smam_bisimulation_exponential_tostructes} ({\em Composition of substitutions}
        for exponential steps), by reading the diagram from the bottom up.
      \item
        Otherwise, note that the reduction step cannot be internal to $\tmthreep$, since \lo\ contexts
        may not go inside substitutions, so it must be internal to $\tmpp$, and the situation is:
          $$
          \commutesredEEABCD{\tostructes}{\tostructes}{
             \tmpp\esub{\vartwo}{\tmthreep}\esub{\var}{\tmfive}
          }{
             \tmtwopp\esub{\vartwo}{\tmthreep}\esub{\var}{\tmfive}
          }{
             \tmpp\esub{\vartwo}{\tmthreep\esub{\var}{\tmfive}}
          }{
             \tmtwopp\esub{\vartwo}{\tmthreep\esub{\var}{\tmfive}}
          }
          $$
        Note that if $\var \not\in \fv{\tmpp}$, then $\var \not\in \fv{\tmtwopp}$,
        by the usual fact that reduction does not create free variables.
      \end{enumerate}
    \item \caselight{$\tostructes$ is applied from right to left.}
      Then $\tmfive$ must be of the form $\tmfivep\esub{\vartwo}{\tmthreep}$,
      and the reduction step must be internal to $\tmp$, so the situation is:
      $$
      \commutesredEEABCD{\tostructes}{\tostructes}{
         \tmp\esub{\var}{\tmfivep\esub{\vartwo}{\tmthreep}}
      }{
         \tmtwop\esub{\var}{\tmfivep\esub{\vartwo}{\tmthreep}}
      }{
         \tmp\esub{\var}{\tmfivep}\esub{\vartwo}{\tmthreep}
      }{
         \tmtwop\esub{\var}{\tmfivep}\esub{\vartwo}{\tmthreep}
      }
      $$
    \end{enumerate}
  \item \casealt{$\tostructdup$ step}
    Two cases, depending on whether the $\tostructdup$ step is applied from left to right
    or from right to left:
    \begin{enumerate}
    \item \caselight{$\tostructdup$ is applied from left to right.}
      \label{l:smam_bisimulation_substitution_tostructdup_lr}
      Then the reduction step is internal to $\tmp$ and closing the diagram is immediate:
        $$
        \commutesredEEABCD{\tostructdup}{\tostructdup}{
           \tmp\esub{\var}{\tmfive}
        }{
           \tmtwop\esub{\var}{\tmfive}
        }{
           \varsplit{\tmp}{\var}{\vartwo}\esub{\var}{\tmfive}\esub{\vartwo}{\tmfive}
        }{
           \varsplit{\tmtwop}{\var}{\vartwo}\esub{\var}{\tmfive}\esub{\vartwo}{\tmfive}
        }
        $$
    \item \caselight{$\tostructdup$ is applied from right to left.}
      Then $\tmp$ must be of the form $\tmpp\esub{\vartwo}{\tmfive}$.
      We consider two further subcases, depending on whether the commuted substitution
      is involved in the reduction step:
      \begin{enumerate}
      \item
        If the reduction step $\tmpp\esub{\vartwo}{\tmfive} \to \tmfourp$ is an exponential step
        and the affected substitution $\esub{\vartwo}{\tmfive}$ is also the one contracted by the
        exponential step, then $\tmpp$ must be of the form $\varsplit{\ctxtwo}{\vartwo}{\var}\ctxholep{\vartwo}$
        and the situation is:
          $$
          \commuteslsEEABCD{\tostructdup}{\tostructdup}{
             \varsplit{\ctxtwo}{\vartwo}{\var}\ctxholep{\vartwo}\esub{\vartwo}{\tmfive}\esub{\var}{\tmfive}
          }{
             \varsplit{\ctxtwo}{\vartwo}{\var}\ctxholep{\tmfive}\esub{\vartwo}{\tmfive}\esub{\var}{\tmfive}
          }{
             \ctxtwo\ctxholep{\vartwo}\esub{\vartwo}{\tmfive}
          }{
             \ctxtwo\ctxholep{\tmfive}\esub{\vartwo}{\tmfive}
          }
          $$
      \item
        Otherwise, note that the reduction step cannot be internal to $\tmfive$, since \lo\ contexts
        may not go inside substitutions, so it must be internal to $\tmpp$.
        The situation is then exactly like in case \reflemma{smam_bisimulation_substitution_tostructdup_lr},
        by reading the diagram from the bottom up.
      \end{enumerate}
    \end{enumerate}
  \end{enumerate}\qed
\end{enumerate}
}
\end{proof}